\documentclass[journal]{IEEEtran}

\usepackage{graphicx}
\usepackage{subfig}
\usepackage{amsmath}
\usepackage{amsthm} 
\usepackage{amssymb}
\usepackage[font=footnotesize]{caption} 
\usepackage{subfig} 
\usepackage[noadjust]{cite} 
\usepackage{color}
\usepackage{algorithm} 
\usepackage{algpseudocode} 
\usepackage{enumerate}
\usepackage[hidelinks,colorlinks=false]{hyperref}
\usepackage{setspace}

\usepackage{tikz}
\usetikzlibrary{calc} 
\usetikzlibrary{shapes} 
\usetikzlibrary{chains}
\usetikzlibrary{fit}
\usetikzlibrary{arrows}
\usetikzlibrary{decorations.text} 
\usetikzlibrary{decorations.markings}

\newtheorem{theorem}{Theorem}
\newtheorem{lemma}{Lemma}
\newtheorem{assumption}{Assumption}
\newtheorem{remark}{Remark}

\newtheorem{definition}{Definition}
\newtheorem{problem}{Problem}

\newcommand{\Null}{\mathrm{Null}}

\newcommand{\one}{\mathbf{1}}
\newcommand{\rank}{\mathrm{rank}}
\newcommand{\myspan}{\mathrm{span}}

\newcommand{\D}{\mathrm{d}}

\newcommand{\sign}{\mathrm{sign}}

\newcommand{\R}{\mathbb{R}}
\newcommand{\G}{\mathcal{G}}
\newcommand{\E}{\mathcal{E}}
\newcommand{\V}{\mathcal{V}}
\newcommand{\N}{\mathcal{N}}

\newcommand{\sat}{\mathrm{sat}}

\renewcommand{\Re}{\mathrm{Re}} 
\renewcommand{\L}{\mathcal{B}}

\graphicspath{{figures/}}

\definecolor{myColor}{RGB}{0,0,255}

\begin{document}

\title{Translational and Scaling Formation Maneuver Control via a Bearing-Based Approach}
\author{Shiyu Zhao and Daniel Zelazo
\thanks{S. Zhao is with the Department of Mechanical Engineering, University of California, Riverside and D. Zelazo is with the Faculty of Aerospace Engineering, Technion - Israel Institute of Technology, Haifa, Israel.
    {\tt\small shiyuzhao@engr.ucr.edu, dzelazo@technion.ac.il}}
}

\maketitle
\begin{abstract}
This paper studies distributed maneuver control of multi-agent formations in arbitrary dimensions.
The objective is to control the translation and scale of the formation while maintaining the desired formation pattern.
Unlike conventional approaches where the target formation is defined by relative positions or distances, we propose a novel bearing-based approach where the target formation is defined by inter-neighbor bearings.
Since the bearings are invariant to the translation and scale of the formation, the bearing-based approach provides a simple solution to the problem of translational and scaling formation maneuver control.
Linear formation control laws for double-integrator dynamics are proposed and the global formation stability is analyzed.
This paper also studies bearing-based formation control in the presence of practical problems including input disturbances, acceleration saturation, and collision avoidance.
The theoretical results are illustrated with numerical simulations.
\end{abstract}

\section{Introduction}

Existing approaches to multi-agent formation control can be categorized by how the desired geometric pattern of the target formation is defined.
In two popular approaches, the target formation is defined by inter-neighbor relative positions or distances (see \cite{Oh2015Automatica} for an overview).
It is notable that the invariance of the constraints of the target formation has an important impact on the formation maneuverability.
For example, since the relative-position constraints are invariant to the translation of the formation, the relative-position-based approach can be applied to realize translational formation maneuvers (see, for example, \cite{RenWei2007SCL}).
Since distance constraints are invariant to both translation and rotation of the formation, the distance-based approach can be applied to realize translational and rotational formation maneuvers (see, for example, \cite{SunZhiyongMSC2015}).

In addition to the above two approaches, there has been a growing research interest in a bearing-based formation control approach in recent years
\cite{bishopconf2011rigid,Eren2012IJC,zhao2015ECC,zelazo2014SE2Rigidity}.
In the bearing-based approach, the geometric pattern of the target formation is defined by inter-neighbor bearings.
Since the bearings are invariant to the translation and scale of the formation, the bearing-based approach provides a simple solution to the problem of translational and scaling formation maneuver control.
Translational maneuvers refer to when the agents move at a common velocity such that the formation translates as a rigid body.
Scaling maneuvers  refer to when the formation scale, which is defined as the average distance from the agents to the formation centroid, varies while the geometric pattern of the formation is preserved.
It is worth mentioning that the bearing-based formation control studied in this paper requires relative-position or velocity measurements, which differs from the bearing-only formation control problem where the feedback control relies only on bearing measurements \cite{nima2009TR,TronVisionFormation,bishop2010SCL,Franchi2012IJRR,Cornejo2013IJRR,ZhengRonghao2013SCL,zhao2013SCLDistribued,Eric2014ACC,zhao2014TACBearing}.
Moreover, bearing-based formation control is a linear control problem whereas bearing-only formation control is nonlinear.

Formation scale control is a useful technique in practical formation control tasks. By adjusting the scale of a formation, a team of agents can dynamically respond to their surrounding environment to, for example, avoid obstacles.
The problem of formation scale control has been studied by the relative-position and distance-based approaches in \cite{Coogan2012Scale,Park2014IJRNC}.
However, since neither the relative positions nor distances are invariant to the formation scale, these two approaches result in complicated estimation and control schemes in which follower agents must estimate the desired formation scale known only by leader agents.
Moreover, the two approaches are so far only applicable in the case where the desired formation scale is constant.
Very recently, the work \cite{LinZhiyun2014TAC} proposed a formation control approach based on the complex Laplacian matrix.
In this approach, the target formation is defined by complex linear constraints that are invariant to the translation, rotation, and scale of the formation.
As a result, this approach provides a simple solution to formation scale control. However, as shown in \cite{LinZhiyun2014TAC}, the approach is only applicable to formation control in the plane; it is unclear if it can be extended to higher dimensions.

Although the bearing-based approach provides a simple solution to formation scale control, the existing studies on bearing-based formation control focus mainly on the case of {static} target formations. The case where the translation and scale of the target formation are time-varying has not yet been studied.
Moreover, a fundamental problem, which has not been solved in the existing literature, is when the target formation can be \emph{uniquely} determined by the inter-neighbor bearings and leaders in arbitrary dimensional spaces.
The analysis of this fundamental problem requires the bearing rigidity theory proposed in \cite{zhao2014TACBearing} and was addressed in our recent work in \cite{zhao2015NetLocalization}.
Our previous work \cite{zhao2015MSC} considered a single-integrator dynamic model of the agents and proposed a proportional-integral bearing-based formation maneuver control law.

The contributions of this paper are summarized as below.
Firstly, we study the problem when a target formation can be uniquely determined by inter-neighbor bearings and leader agents.
The necessary and sufficient condition for uniqueness of the target formation is analyzed based on a special matrix we term the \emph{bearing Laplacian}, which characterizes both the interconnection topology and the inter-neighbor bearings of the formation.
Secondly, we propose two linear bearing-based formation control laws for double-integrator dynamics. With these two control laws, the formation can track constant or time-varying leader velocities.
In the proposed control laws, the desired translational and scaling maneuver is only known to the leaders and the followers are not required to estimate it.
A global formation stability analysis is presented for each of the control laws.
Thirdly, we study bearing-based formation control in the presence of some practical issues.
In particular, control laws that can handle constant input disturbances and acceleration saturation are proposed and their global stability is analyzed.
Sufficient conditions that ensure no collision between any two agents are also proposed.
Finally, it is noteworthy that the results presented in this paper are applicable to formation control in arbitrary dimensions.

The organization of this paper is as follows.
Section~\ref{section_problemStatement} presents the problem formulation.
Section~\ref{section_controlLaws} proposes and analyzes two linear bearing-based formation control laws.
Section~\ref{section_practicalIssues} considers bearing-based formation control in the presence of practical issues such as input disturbances and acceleration saturation.
Conclusions are drawn in Section~\ref{section_conclusion}.

\section{The Formation Maneuver Control Problem\\ and Bearing-Constrained Target Formations}\label{section_problemStatement}

Consider a formation of $n$ agents in $\R^d$ ($n\ge2$, $d\ge2$).
Let $\V\triangleq\{1,\dots,n\}$.
Denote $p_i(t)\in\R^d$ and $v_i(t) \in \R^d$ as the position and velocity of agent $i\in\V$.
Let the first $n_\ell$ agents be termed the \emph{leaders} and the remaining $n_f$ agents the \emph{followers} ($n_\ell+n_f=n$).
Let $\V_\ell=\{1,\dots,n_\ell\}$ and $\V_f=\{n_\ell+1,\dots,n\}$ be the index sets of the leaders and followers, respectively.
The motion (i.e., position and velocity) of each leader is {given} \emph{a priori}, and we assume the velocity of each leader is piecewise continuously differentiable.
Each follower is modeled as a double-integrator,
\begin{align*}
\dot{p}_i(t)=v_i(t),\quad \dot{v}_i(t)=u_i(t), \quad i\in\V_f,
\end{align*}
where $u_i(t)\in\R^d$ is the acceleration input to be designed.
Let $p_\ell=[p_1^T ,\dots, p_{n_\ell}^T ]^T$, $p_f=[p_{n_\ell+1}^T  ,\dots, p_{n}^T ]^T $, $v_\ell=[v_1^T  ,\dots, v_{n_\ell}^T ]^T $, and $v_f=[v_{n_\ell+1}^T  ,\dots, v_n^T]^T $.
Let $p=[p_\ell^T,p_f^T]^T$ and $v=[v_\ell^T,v_f^T]^T$.

The underlying information flow among the agents is described by a fixed graph $\G=(\V,\E)$ where $\E\subset\V\times \V$ is the edge set.
By mapping the point $p_i$ to the vertex $i$, we denote the formation as $\G(p)$.
If $(i,j)\in\E$, agent $i$ can access to the information of agent $j$.
The set of neighbors of agent $i$ is denoted as $\mathcal{N}_i\triangleq\{j \in \mathcal{V}: (i,j)\in \mathcal{E}\}$.
We assume that the information flow between any two followers is \emph{bidirectional}.
The bearing of agent $j$ relative to agent $i$ is described by the unit vector
\begin{align*}
g_{ij}\triangleq \frac{p_j-p_i}{\|p_j-p_i\|}.
\end{align*}
Note $g_{ji}=-g_{ij}$.
For $g_{ij}$, define
\begin{align*}
    P_{g_{ij}} \triangleq I_d - g_{ij}g_{ij}^T,
\end{align*}
where $I_d\in\R^{d\times d}$ is the identity matrix.
Note that $P_{g_{ij}}$ is an orthogonal projection matrix that geometrically projects any vector onto the orthogonal compliment of $P_{g_{ij}}$.
It can be verified that $P_{g_{ij}}$ is positive semi-definite and satisfies $P_{g_{ij}}^T =P_{g_{ij}}$, $P_{g_{ij}}^2=P_{g_{ij}}$, and $\Null(P_{g_{ij}})=\myspan\{g_{ij}\}$.

\subsection{Bearing-Based Formation Maneuver Control}

Suppose the real bearings of the formation at time $t>0$ are $\{g_{ij}(t)\}_{(i,j)\in\E}$, and the desired constant bearings are $\{g_{ij}^*\}_{(i,j)\in\E}$.
The bearing-based formation control problem is formally stated below.

\begin{problem}[{Bearing-Based Formation Maneuver Control}]\label{problem_bearingFormationManeuver}
Consider a formation $\G(p(t))$ where the (time-varying) position and velocity of the leaders, $\{p_i(t)\}_{i\in\V_\ell}$ and $\{v_i(t)\}_{i\in\V_\ell}$, are given.
Design the acceleration control input $u_i(t)$ for each follower $i\in\V_f$ based on the relative position $\{p_i(t)-p_j(t)\}_{j\in\N_i}$ and the relative velocity $\{v_i(t)-v_j(t)\}_{j\in\N_i}$ such that $g_{ij}(t)\rightarrow g_{ij}^*$ for all $(i,j)\in\E$ as $t\rightarrow \infty$.
\end{problem}

Problem~\ref{problem_bearingFormationManeuver} can be equivalently stated as a problem where the formation is required to converge to a bearing-constrained target formation as defined below.

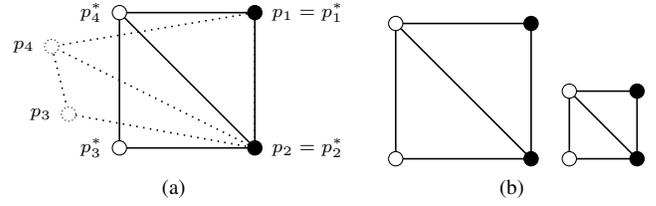
\begin{figure}
  \centering
  \def\myscale{0.45}
\subfloat[]{
\begin{tikzpicture}[scale=\myscale]
            \coordinate (x1) at (0,0);
            \coordinate (x2) at (0,-4);
            \coordinate (x3_star) at (-4,-4);
            \coordinate (x4_star) at (-4,0);
            \coordinate (x3) at (-5.5,-3);
            \coordinate (x4) at (-6,-1);
            \draw [semithick] (x1)--(x2)--(x3_star)--(x4_star)--cycle;
            \draw [semithick] (x2)--(x4_star);
            \draw [semithick,dotted] (x1)--(x2)--(x3)--(x4)--cycle;
            \draw [semithick,dotted] (x2)--(x4);
            \def\radius{6pt}
            \draw [fill=black](x1) circle [radius=\radius];
            \draw [fill=black](x2) circle [radius=\radius];
            \draw [fill=white](x3_star) circle [radius=\radius];
            \draw [fill=white](x4_star) circle [radius=\radius];
            \draw [fill=white,densely dotted](x3) circle [radius=\radius];
            \draw [fill=white,densely dotted](x4) circle [radius=\radius];
            \draw (x1) node[right=\radius/2] {\scriptsize{$p_1=p_1^*$}};
            \draw (x2) node[right=\radius/2] {\scriptsize$p_2=p_2^*$};
            \draw (x3_star) node[left=\radius/2] {\scriptsize$p_3^*$};
            \draw (x4_star) node[left=\radius/2] {\scriptsize$p_4^*$};
            \draw (x3) node[left=\radius/2] {\scriptsize$p_3$};
            \draw (x4) node[left=\radius/2] {\scriptsize$p_4$};
\end{tikzpicture}}
\quad
\subfloat[]{
\begin{tikzpicture}[scale=\myscale]
\def\r{4}
            \coordinate (x1) at (0,0);
            \coordinate (x2) at (0,-\r);
            \coordinate (x3) at (-\r,-\r);
            \coordinate (x4) at (-\r,0);
            \draw [semithick] (x1)--(x2)--(x3)--(x4)--cycle;
            \draw [semithick] (x2)--(x4);
            \def\radius{6pt}
            \draw [fill=black](x1) circle [radius=\radius];
            \draw [fill=black](x2) circle [radius=\radius];
            \draw [fill=white](x3) circle [radius=\radius];
            \draw [fill=white](x4) circle [radius=\radius];
\end{tikzpicture}
\,
\begin{tikzpicture}[scale=\myscale]
\def\r{2}
            \coordinate (x1) at (0,0);
            \coordinate (x2) at (0,-\r);
            \coordinate (x3) at (-\r,-\r);
            \coordinate (x4) at (-\r,0);
            \draw [semithick] (x1)--(x2)--(x3)--(x4)--cycle;
            \draw [semithick] (x2)--(x4);
            \def\radius{6pt}
            \draw [fill=black](x1) circle [radius=\radius];
            \draw [fill=black](x2) circle [radius=\radius];
            \draw [fill=white](x3) circle [radius=\radius];
            \draw [fill=white](x4) circle [radius=\radius];
\end{tikzpicture}} 
  \caption{An illustration of the bearing-constrained target formation. Solid dots: leaders; hollow dots: followers. Figure (a) shows the target formation $p^*$ and the real formation $p$. Figure (b) shows two target formations that have the same bearings but different translations and scales.}
  \label{fig_illustration_targetFormation}
\end{figure}

\begin{definition}[Target Formation]\label{definition_targetFormation}
The \emph{target formation} denoted by $\G(p^*(t))$ is a formation that satisfies the following constraints for all $t\ge0$:
\begin{enumerate}[(a)]
\item \hspace{-5pt} Bearing: $(p_j^*(t)-p_i^*(t))/\|p_j^*(t)-p_i^*(t)\|=g_{ij}^*, \forall (i,j)\in\E$,
\item \hspace{-5pt} Leader: $p_i^*(t)=p_i(t), \forall i\in\V_\ell$.
\end{enumerate}
\end{definition}

The target formation $\G(p^*(t))$ is constrained jointly by the bearing constraints and the leader positions.
The bearing constraints are constant, but the leader positions may be time-varying.
Given appropriate motion of the leaders, the target formation has the desired translational and scaling maneuver and desired inter-neighbor bearings.
If the real formation $p(t)$ converges to the target formation $p^*(t)$, the desired formation maneuver and formation pattern can be simultaneously achieved.
Motivated by this idea, define the position and velocity errors for the followers as
\begin{align}\label{eq_trackingErrorDef}
\delta_p(t) = p_f(t)-p_f^*(t),\quad
\delta_v(t) = v_f(t)-v_f^*(t),
\end{align}
where $p_f^*(t)$ and $v_f^*(t)$ are the position and velocity of the followers in the target formation.
The control objective is to design control laws for the followers to drive $\delta_p(t)\rightarrow 0$ and $\delta_v(t)\rightarrow0$ as $t\rightarrow\infty$ (see Figure~\ref{fig_illustration_targetFormation} for an illustration).
Note $\dot{\delta}_p(t)=\delta_v(t)$.

A fundamental problem regarding the target formation, which is still unexplored so far, is whether or not $p^*(t)$ exists and is unique.
If $p^*(t)$ is not unique, there exist multiple formations satisfying the bearing constraints and leader positions, and consequently the formation may not be able to converge to the desired geometric pattern.
This fundamental problem is analyzed in the following subsection.

\subsection{Properties of the Target Formation}

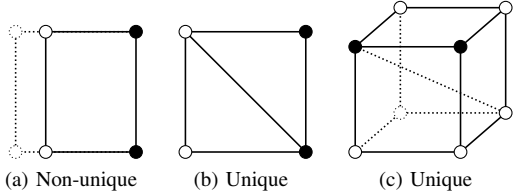
\begin{figure}[!t]
  \centering
  \def\myscale{0.4}
\subfloat[Non-unique]{
\begin{tikzpicture}[scale=\myscale]
            \coordinate (x1) at (0,0);
            \coordinate (x2) at (0,-4);
            \coordinate (x3) at (-3,-4);
            \coordinate (x4) at (-3,0);
            \coordinate (x5) at (-4,0);
            \coordinate (x6) at (-4,-4);
            \draw [semithick] (x1)--(x2)--(x3)--(x4)--cycle;
            \draw [semithick,densely dotted] (x1)--(x5)--(x6)--(x2);
            \def\radius{6pt}
            \draw [fill=black](x1) circle [radius=\radius];
            \draw [fill=black](x2) circle [radius=\radius];
            \draw [fill=white](x3) circle [radius=\radius];
            \draw [fill=white](x4) circle [radius=\radius];
\draw [fill=white,densely dotted](x5) circle [radius=\radius];
\draw [fill=white,densely dotted](x6) circle [radius=\radius];
\end{tikzpicture}}
\quad
\subfloat[Unique]{
\begin{tikzpicture}[scale=\myscale]
            \coordinate (x1) at (0,0);
            \coordinate (x2) at (0,-4);
            \coordinate (x3) at (-4,-4);
            \coordinate (x4) at (-4,0);
            \draw [semithick] (x1)--(x2)--(x3)--(x4)--cycle;
            \draw [semithick] (x2)--(x4);
            \def\radius{6pt}
            \draw [fill=black](x1) circle [radius=\radius];
            \draw [fill=black](x2) circle [radius=\radius];
            \draw [fill=white](x3) circle [radius=\radius];
            \draw [fill=white](x4) circle [radius=\radius];
\end{tikzpicture}}
\quad
\subfloat[Unique]{
\begin{tikzpicture}[scale=\myscale]
            \def\length{3.5}
            \coordinate (x1) at (0,0);
            \coordinate (x2) at (\length,0);
            \coordinate (x3) at (\length,-\length);
            \coordinate (x4) at (0,-\length);
            \def\Xoffset{1.5}
            \def\Yoffset{1.3}
            \coordinate (x5) at (0+\Xoffset,0+\Yoffset);
            \coordinate (x6) at (\length+\Xoffset,0+\Yoffset);
            \coordinate (x7) at (\length+\Xoffset,-\length+\Yoffset);
            \coordinate (x8) at (0+\Xoffset,-\length+\Yoffset);
            \draw [semithick] (x1)--(x2)--(x3)--(x4)--cycle;
            \draw [semithick] (x1)--(x5)--(x6)--(x2);
            \draw [semithick] (x6)--(x7)--(x3);
            \draw [semithick,densely dotted] (x5)--(x8);
            \draw [semithick,densely dotted] (x7)--(x8);
            \draw [semithick,densely dotted] (x4)--(x8);
            \draw [semithick,densely dotted] (x1)--(x7);
            \def\radius{6pt}
            \draw [fill=black](x1) circle [radius=\radius];
            \draw [fill=black](x2) circle [radius=\radius];
            \draw [fill=white](x3) circle [radius=\radius];
            \draw [fill=white](x4) circle [radius=\radius];
            \draw [fill=white](x5) circle [radius=\radius];
            \draw [fill=white](x6) circle [radius=\radius];
            \draw [fill=white](x7) circle [radius=\radius];
            \draw [fill=white,densely dotted](x8) circle [radius=\radius];
\end{tikzpicture}} 
  \caption{Examples of non-unique and unique target formations. Solid dots: leaders; hollow dots: followers. The formation in (c) is three-dimensional.}
  \label{fig_Example_Targetformation}
\end{figure}

This subsection explores the properties of the target formation that will be used throughout the paper.

\subsubsection{Bearing Laplacian Matrix}

Define a matrix $\L(G(p^*))\in\R^{dn\times dn}$ with the $ij$th block of submatrix as
\begin{align*}
[\L(G(p^*))]_{ij} = \left\{
  \begin{array}{ll}
      {\bf 0}_{d \times d}, &i \neq j, \, (i,j)\notin\E, \\
      -P_{g_{ij}^*}, & i \neq j, \, (i,j)\in\E, \\ 
      \sum_{k\in\N_i}P_{g_{ik}^*}, & i=j, \, i\in\V. \\
  \end{array}
\right.
\end{align*}
The matrix $\L(G(p^*))$, which we write in short as $\L$ in the sequel, can be viewed as a matrix-weighted graph Laplacian matrix, where the matrix weight for each edge is a positive semi-definite orthogonal projection matrix.
We call $\L$ the \emph{bearing Laplacian} since it characterizes both the interconnection topology and the bearings of the formation.
The bearing Laplacian matrix naturally emerges and plays important roles in bearing-based formation control and network localization problems \cite{zhao2015ECC,zhao2015MSC,zhao2015NetLocalization}.

We now state an important property of the bearing Laplacian.
In the following, $\one_n\in\R^n$ is the vector with all entries equal to one, and $\otimes$ denotes the Kronecker matrix product.

\begin{lemma}\label{lemma_NetworkLaplacian_nullspace}
For any $\G(p^*)$, the bearing Laplacian always satisfies
\begin{align}\label{eq_nullspaceBearingLaplacian}
\Null(\L)\supseteq\myspan\{\one_n\otimes I_d, p^*\}.
\end{align}
\end{lemma}
\begin{proof}
For any $x=[x_1^T,\dots,x_n^T]^T\in\R^{dn}$, we have
\begin{align}\label{eq_LxElement}
\L x=
\left[
  \begin{array}{c}
    \vdots \\
    \sum_{j\in\N_i} P_{g_{ij}^*}(x_i-x_j) \\
    \vdots \\
  \end{array}
\right].
\end{align}
Firstly, if $x\in\myspan\{\one_n\otimes I_d\}$, then $x_i=x_j$ for all $i,j\in\V$.
It then follows from \eqref{eq_LxElement} that $\L x=0$.
Secondly, if $x\in\myspan\{p^*\}$, then $x_i=kp_i^*$ for all $i\in\V$ where $k\in\R$.
It then follows from $P_{g_{ij}^*}(p_i^*-p_j^*)=0$ that $\L x=0$.
To sum up, any vector in $\myspan\{\one_n\otimes I_d, p^*\}$ is also in $\Null(\L)$.
\end{proof}
\begin{remark}
In fact, any vector in the null space of $\L$ corresponds to a motion of the formation that preserves all the bearings \cite{zhao2015NetLocalization}.
As a result, the expression in \eqref{eq_nullspaceBearingLaplacian} indicates that the bearings are invariant to the translational and scaling motion of the formation.
Specifically, $\one_n\otimes I_d$ corresponds to the translational motion and $p^*-\one_n\otimes (\sum_{i=1}^n p_i^*/n)$ corresponds to the scaling motion.
In addition, the bearings may also be invariant to other bearing-preserving motions (see, for example, Figure~\ref{fig_Example_Targetformation}(a)).
It is of great interest to understand when $\Null(\L)$ exactly equals $\myspan\{\one_n\otimes I_d, p^*\}$.
As shown in \cite{zhao2015NetLocalization}, when $\G$ is undirected, $\Null(\L)=\myspan\{\one_n\otimes I_d, p^*\}$ if and only if $\G(p)$ is infinitesimally bearing rigid.
The definition of the infinitesimal bearing rigidity and preliminaries to the bearing rigidity theory are given in the appendix.
\end{remark}

We continue with the analysis by partitioning $\L$ as
\begin{align*}
    \L=\left[
         \begin{array}{cc}
           \L_{\ell\ell} & \L_{\ell f} \\
           \L_{f\ell} & \L_{ff} \\
         \end{array}
       \right],
\end{align*}
where $\L_{\ell\ell}\in\R^{dn_\ell\times dn_\ell}$, $\L_{\ell f}\in\R^{dn_\ell\times d n_f}$ $\L_{f\ell}\in\R^{dn_f\times d n_\ell} $, and $\L_{ff}\in\R^{dn_f \times dn_f}$.
As will be shown later, the submatrix $\L_{ff}$ plays an important role in this work.

\begin{lemma}
The submatrix $\L_{ff}\in\R^{dn_f\times dn_f}$ is symmetric and positive semi-definite.
\end{lemma}
\begin{proof}

The submatrix $\L_{ff}$ can be written as $\L_{ff}=\L_0+\mathcal{D}$ where $\L_0\in\R^{dn_f\times dn_f}$ is the bearing Laplacian for the subgraph of the followers and $\mathcal{D}\in\R^{dn_f\times dn_f}$ is a positive semi-definite block-diagonal matrix with $[\mathcal{D}]_{ii}=\sum_{j\in\V_\ell\cap\N_i}P_{g_{ij}^*}$ for $i\in\V_f$.
Note $\L_0$ is symmetric because the edges among the followers are assumed to be bidirectional.
For any $x=[x_1^T,\dots,x_{n_f}^T]^T\in\R^{dn_f}$, we have $x^T\L_0 x=\sum_{i\in\V_f}\sum_{j\in\V_f\cap\N_i}\|P_{g_{ij}^*}(x_i-x_j)\|^2\ge0$ and hence $\L_0$ is positive semi-definite.
Since $\mathcal{D}$ is also positive semi-definite, the matrix $\L_{ff}$ is positive semi-definite.
\end{proof}

\subsubsection{Uniqueness of the Target Formation}
Based on the bearing Laplacian, we can analyze the existence and uniqueness of the target formation $p^*$ (i.e., the existence and uniqueness of solutions to the equations in Definition~\ref{definition_targetFormation}).
The bearing constraints and leader positions are feasible if there exists at least one formation that satisfies them.
Feasible bearings and leader positions may be calculated from an arbitrary formation configuration that has the desired geometric pattern.
In general, given a set of feasible bearing constraints and leader positions, the target formation may \emph{not} be unique (see, for example, Figure~\ref{fig_Example_Targetformation}(a)).
In fact, the uniqueness problem of the target formation is identical to the localizability problem in bearing-only network localization \cite{zhao2015NetLocalization}.
We next give the necessary and sufficient condition for uniqueness of the target formation.

\begin{theorem}[Uniqueness of the Target Formation]\label{theorem_uniquenessOfTargetFormation}
Given feasible bearing constraints and leader positions, the target formation in Definition~\ref{definition_targetFormation} is unique if and only if $\L_{ff}$ is nonsingular.
When $\L_{ff}$ is nonsingular, the position and velocity of the followers in the target formation are uniquely determined as
\begin{align}\label{eq_pfstarExpression}
p_f^*(t)=-\L_{ff}^{-1}\L_{f\ell}p_\ell(t), \quad v_f^*(t)=-\L_{ff}^{-1}\L_{f\ell}v_\ell(t).
\end{align}
\end{theorem}
\begin{proof}
As shown in \cite{zhao2015NetLocalization}, the target formation is uniquely determined by the bearings and leader positions if and only if $\L_{ff}$ is nonsingular.
It follows from Lemma~\ref{lemma_NetworkLaplacian_nullspace} that $\L p^*=0$, which further implies $\L_{ff}p_f^*+\L_{f\ell}p_\ell=0$.
When $\L_{ff}$ is nonsingular, $p_f^*=-\L_{ff}^{-1}\L_{f\ell}p_\ell$.
Then, $v_f^*=\dot{p}_f^*=-\L_{ff}^{-1}\L_{f\ell}v_\ell$.
\end{proof}

A variety of other conditions for uniqueness of the target formation can be found in \cite{zhao2015NetLocalization}.
Here we highlight two useful conditions.
A useful necessary condition is that a unique target formation must have at least two leaders.
In this paper, we always assume there exist at least two leaders.
A useful sufficient condition is that the target formation is unique if it is infinitesimally bearing rigid and has at least two leaders.
By the sufficient condition, in order to design a unique target formation, we can first design an infinitesimally bearing rigid formation and then arbitrarily assign two agents as leaders.
Figure~\ref{fig_Example_Targetformation}(b)--(c) shows examples of unique target formations (more examples can be found in \cite{zhao2015NetLocalization}).
For the analysis in the sequel, we adopt the following uniqueness assumption.

\begin{assumption}
The target formation $\G(p^*(t))$ is unique for all $t\ge0$, which means $\L_{ff}$ is nonsingular.
\end{assumption}

\subsubsection{Target Formation Maneuvering}

In bearing-based formation maneuver control, the desired translational and scaling maneuver of the formation is known only to the leaders.
In order to achieve the desired maneuvers, the leaders must have appropriate motions.
We now study \emph{how} the leaders should move to achieve the desired maneuvers of the target formation.
Formation control laws will be designed later such that the real formation is steered to track the target formation.

To describe the translational and scaling maneuvers, we define the \emph{centroid}, $c(p^*(t))$, and the \emph{scale}, $s(p^*(t))$, for the target formation as
\begin{align*}
c(p^*(t))&\triangleq\frac{1}{n}\sum_{i\in\V}p_i^*(t)=\frac{1}{n}(\one_n\otimes I_d)^T {p}^*(t),\\
s(p^*(t))&\triangleq\sqrt{\frac{1}{n}\sum_{i\in\V}\|p_i^*(t)-c(p^*(t))\|^2}\\
&=\frac{1}{\sqrt{n}}\|p^*(t)-\one_n\otimes c(p^*(t))\|.
\end{align*}
The desired maneuvering dynamics of the centroid and scale of the target formation are given by
\begin{align}\label{eq_centroldScaleDesiredDynamics}
\dot{c}(p^*(t))=v_c(t), \quad
\dot{s}(p^*(t))=\alpha(t)s(p^*(t)),
\end{align}
where $v_c(t)\in\R^d$ denotes the desired velocity common to all agents and $\alpha(t)\in\R$ is the varying rate of the scale.
The formation scale expands when $\alpha(t)>0$ and contracts when $\alpha(t)<0$.
Suppose $v_c(t)$ and $\alpha(t)$ are known by the leaders.
We next show how the leaders should move to achieve the desired dynamics in \eqref{eq_centroldScaleDesiredDynamics}.

\begin{theorem}[Target Formation Maneuvering]\label{theorem_bothTransAndScale}
The desired dynamics of the centroid and the scale given in \eqref{eq_centroldScaleDesiredDynamics} are achieved if the velocities of the leaders have the form of
\begin{align}\label{eq_leaderVelocity}
v_i(t)=v_c(t)+\alpha(t)[p_i(t)- c(p^*(t))], \quad i\in\V_\ell.
\end{align}
\end{theorem}
\begin{proof}
The vector form of \eqref{eq_leaderVelocity} is $v_\ell(t)=\one_{n_\ell}\otimes v_c(t)+\alpha(t)[p_\ell(t)- \one_{n_\ell}\otimes c(p^*(t))]$.
Since $\myspan\{\one_n\otimes I_d, p^*\}\subseteq\Null(\L)$ as given in Lemma~\ref{lemma_NetworkLaplacian_nullspace}, we have
\begin{align*}
\L\left(\one_{n}\otimes v_c(t)+\alpha(t)(p^*(t)- \one_{n}\otimes c(p^*(t)))\right)=0.
\end{align*}
The above equation implies $\L_{f\ell}v_\ell+\L_{ff}[\one_{n_f}\otimes v_c(t)+\alpha(t)[p_f^*(t)- \one_{n_f}\otimes c(p^*(t)))]=0$.
Then $v_f^*(t)$ is calculated as
\begin{align*}
v_f^*(t)
&=\L_{ff}^{-1}\L_{f\ell}v_\ell(t)\\
&=\one_{n_f}\otimes v_c(t)+\alpha(t)[p_f^*(t)- \one_{n_f}\otimes c(p^*(t))],
\end{align*}
whose elementwise form is $v_i^*(t)=v_c(t)+\alpha(t)(p_i^*(t)-c(p^*(t)))$ for all $i\in\V_f$.
Note $\dot{p}^*=[v_{\ell}^T ,(v_f^*)^T ]^T =\one_n\otimes v_c(t)+\alpha(t)(p^*-\one_{n}\otimes c(p^*))$.
Substituting $\dot{p}^*$ into $\dot{c}(p^*)$ and $\dot{s}(p^*)$ gives
\begin{align*}
\dot{c}(p^*)&=\frac{1}{n}(\one_n\otimes I_d)^T \dot{p}^*\\
&=\frac{1}{n}(\one_n\otimes I_d)^T [\one_n\otimes v_c(t)+\alpha(t)(p^*-\one_{n}\otimes c(p^*))] \\
&=\frac{1}{n}(\one_n\otimes I_d)^T (\one_n\otimes v_c(t))
=v_c(t),
\end{align*}
and
\begin{align*}
\dot{s}(p^*)&=\frac{1}{\sqrt{n}}\frac{(p^*-\one_n\otimes c(p^*))^T }{\|p^*-\one_n\otimes
c(p^*)\|}(\dot{p}^*-\one_n\otimes v_c(t))\\
&=\frac{1}{\sqrt{n}}\frac{(p^*-\one_n\otimes c(p^*))^T }{\|p^*-\one_n\otimes c(p^*)\|}\alpha(t)(p^*-\one_{n}\otimes c(p^*))\\
&=\alpha(t) s(p^*).
\end{align*}
\end{proof}

As shown in \eqref{eq_leaderVelocity}, the velocity of each leader should be a linear combination of the common translational velocity and the velocity induced by the scaling variation.
In addition to $v_c(t)$ and $\alpha(t)$, each leader should also know the centroid $c(p^*(t))$, which is a global information of the target formation.
This quantity may be estimated in a distributed way using, for example, consensus filters, as described in \cite{Yang2010Automatica}.

\section{Bearing-Based Formation Control Laws}\label{section_controlLaws}
In this section, we propose two distributed control laws to steer the followers to track the maneuvering target formation.
The first control law requires relative position and velocity feedback; with this control law the formation tracks target formations with constant velocities.
The second control law requires position, velocity, and acceleration feedback; with this control law the formation tracks target formations with time-varying velocities.

\subsection{Formation Maneuvering with Constant Leader Velocity}

The bearing-based control law for follower $i\in\V_f$ is proposed as
\begin{align}\label{eq_controlLaw_ConstantVelocity}
u_i=-\sum_{j\in\N_i}P_{g_{ij}^*}\left[k_p(p_i-p_j)+k_v(v_i-v_j)\right],
\end{align}
where $P_{g_{ij}^*}=I_d-g_{ij}^*(g_{ij}^*)^T$ is a constant orthogonal projection matrix, and $k_p$ and $k_v$ are positive constant control gains.
Several remarks on the control law are given below.
Firstly, the neighbor $j\in\N_i$ of agent $i$ may be either a follower or a leader.
Secondly, the proposed control law has a clear geometric meaning illustrated in Figure~\ref{fig_controlLawGeometricMeaning}: the control term $P_{g_{ij}^*}(p_j-p_i)$ steers agent $i$ to a position where $g_{ij}$ is aligned with $g_{ij}^*$.
Thirdly, the proposed control law has a similar form as the second-order linear consensus protocols \cite{RenWei2008TAC,YuWenwu2010Automatica,MengZiyang2013SCL}.
The difference is that, in the consensus protocols the weight for each edge is a positive scalar, whereas in the proposed control law the weight for each edge is a positive semi-definite orthogonal projection matrix.
It is precisely the special properties of the projection matrices that allows the proposed control law to solve the bearing-based formation control problem.
The convergence of control law \eqref{eq_controlLaw_ConstantVelocity} is analyzed below.

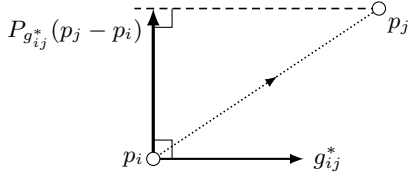
\begin{figure}
  \centering
  \def\myscale{0.5}
  \tikzset{->-/.style={decoration={
  markings,
  mark=at position #1 with {\arrow{>}}},postaction={decorate}}
}

\begin{tikzpicture}[scale=\myscale]
            \coordinate (xi) at (0,0);
            \coordinate (xj) at (6,4);
            \coordinate (proj) at (0,4);
            \coordinate (projLeft) at (-0.5,4);
            \def\unit{4}
            \coordinate (gij)                 at (\unit,0);
            \draw [densely dotted,->-=0.55,>=latex, semithick] (xi)--(xj);
            \draw [densely dashed, semithick] (xj)--(projLeft);
            \draw[->, >=latex, thick] (xi) -- (gij) node [right] {\small{$g_{ij}^*$}};
            \draw[->, >=latex, very thick] (xi) -- (proj) node [below left] {\small{$P_{{g}_{ij}^*}({p}_j-{p}_i)$}};
            \def\radius{5pt}
            \draw [fill=white](xi) circle [radius=\radius];
            \draw [fill=white](xj) circle [radius=\radius];
            \draw (xi) node[left] {\small{${p}_i$}};
            \draw (xj) node[below right] {\small{${p}_j$}};
            \def\dis{0.5}
            \coordinate (p1) at (0,\dis);
            \coordinate (p2) at (\dis,\dis);
            \coordinate (p3) at (\dis,0);
            \draw [] (p1)--(p2)--(p3);
            \coordinate (p4) at (0,4-\dis);
            \coordinate (p5) at (\dis,4-\dis);
            \coordinate (p6) at (\dis,4);
            \draw [] (p4)--(p5)--(p6);
\end{tikzpicture}
  \caption{The geometric meaning of the term $P_{g_{ij}^*}(p_j-p_i)$ in control law \eqref{eq_controlLaw_ConstantVelocity}.}
  \label{fig_controlLawGeometricMeaning}
\end{figure}

\begin{theorem}\label{theorem_convergenceConstantVelocity}
Under control law \eqref{eq_controlLaw_ConstantVelocity}, when the leader velocity $v_\ell(t)$ is constant, the tracking errors $\delta_p(t)$ and $\delta_v(t)$ as defined in \eqref{eq_trackingErrorDef} globally and exponentially converge to zero.
\end{theorem}
\begin{proof}
With control law \eqref{eq_controlLaw_ConstantVelocity}, the dynamics of the followers can be expressed in a matrix-vector form as
\begin{align}\label{eq_controlLaw_ConstantVelocity_matrixForm}
\dot{v}_f
&=-k_p(\L_{ff}p_f+\L_{f\ell}p_\ell)-k_v(\L_{ff}v_f+\L_{f\ell}v_\ell) \nonumber\\
&=-k_p\L_{ff}\delta_p-k_v\L_{ff}\delta_v,
\end{align}
where the second equality is due to the fact that $\delta_p=p_f+\L_{ff}^{-1}\L_{f\ell}p_\ell$ and $\delta_v=v_f+\L_{ff}^{-1}\L_{f\ell}v_\ell$ as shown in \eqref{eq_pfstarExpression}.
Substituting \eqref{eq_controlLaw_ConstantVelocity_matrixForm} into the error dynamics gives $\dot{\delta}_v=\dot{v}_f+\L_{ff}^{-1}\L_{f\ell}\dot{v}_\ell=-k_p\L_{ff}\delta_p-k_v\L_{ff}\delta_v+\L_{ff}^{-1}\L_{f\ell}\dot{v}_\ell$,
which can be rewritten in a compact form as
\begin{align}\label{eq_errorDynamics_constantVelocity_matrixForm}
\left[
  \begin{array}{c}
    \dot{\delta}_p \\
    \dot{\delta}_v \\
  \end{array}
\right]
=
\left[
  \begin{array}{cc}
    0 & I \\
    -k_p\L_{ff} & -k_v\L_{ff} \\
  \end{array}
\right]
\left[
  \begin{array}{c}
    \delta_p \\
    \delta_v \\
  \end{array}
\right]
+\left[
  \begin{array}{c}
    0 \\
    \L_{ff}^{-1}\L_{f\ell} \\
  \end{array}
\right]\dot{v}_\ell.
\end{align}
Let $\lambda$ be an eigenvalue of the state matrix of \eqref{eq_errorDynamics_constantVelocity_matrixForm}.
The characteristic equation of the state matrix is given by $\det(\lambda^2I+\lambda k_v\L_{ff}+k_p\L_{ff})=0$.
It can be calculated that $\lambda=(-k_v\mu\pm \sqrt{k_v^2\mu^2-4k_p\mu})/2$, where $\mu>0$ is an eigenvalue of $\L_{ff}$.
Therefore, $\Re(\lambda)<0$ for any $k_p, k_v, \mu>0$.
As a result, the state matrix is Hurwitz and hence $\delta_p$ and $\delta_v$ globally and exponentially converge to zero when $\dot{v}_\ell\equiv0$.
\end{proof}

\begin{figure}
  \centering
  \subfloat[Trajectory]{\includegraphics[width=\linewidth]{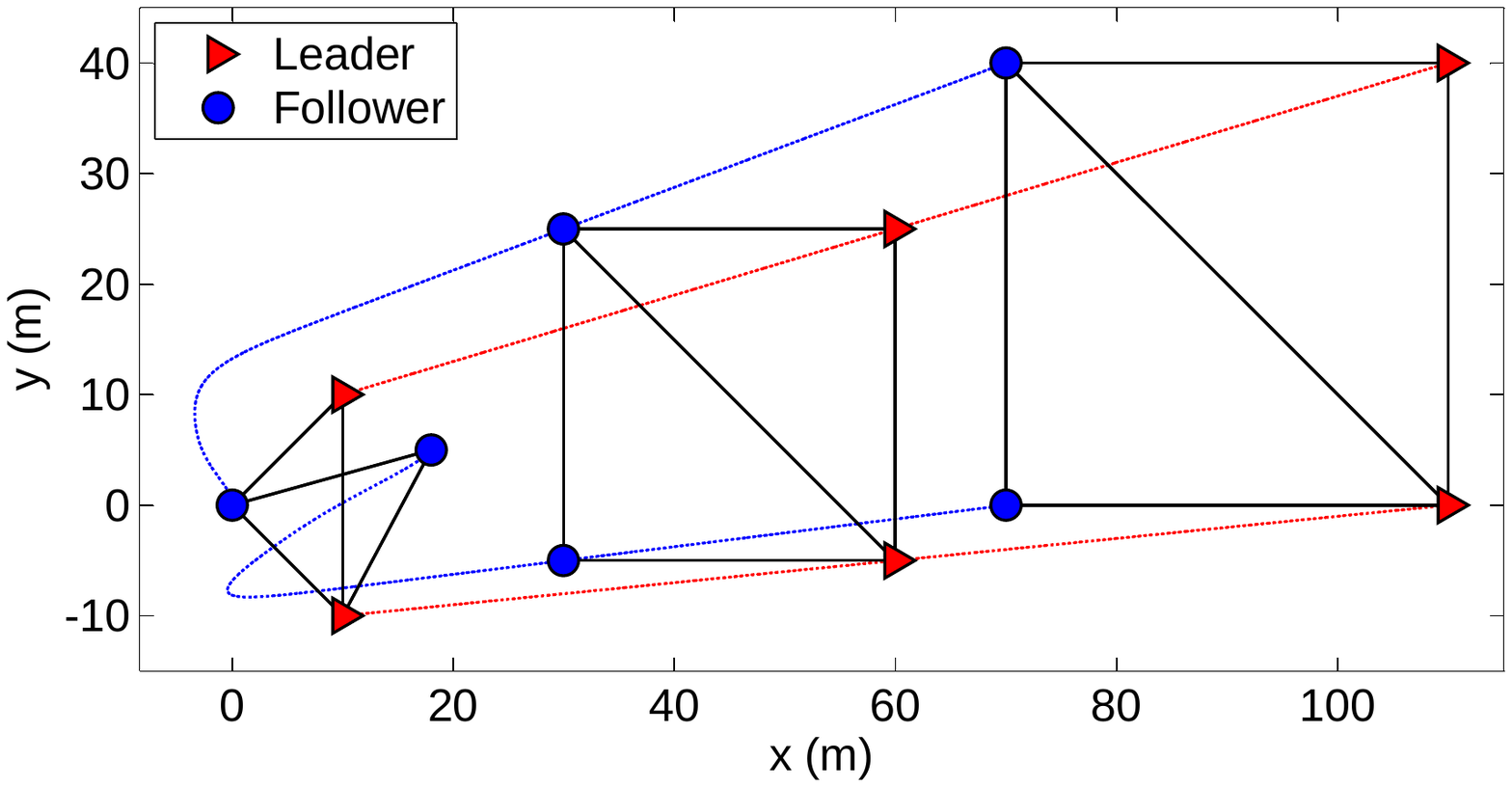}}\\
  \subfloat[Total bearing error: $\sum_{(i,j)\in\E}\|g_{ij}(t)-g_{ij}^*\|$]{\includegraphics[width=\linewidth]{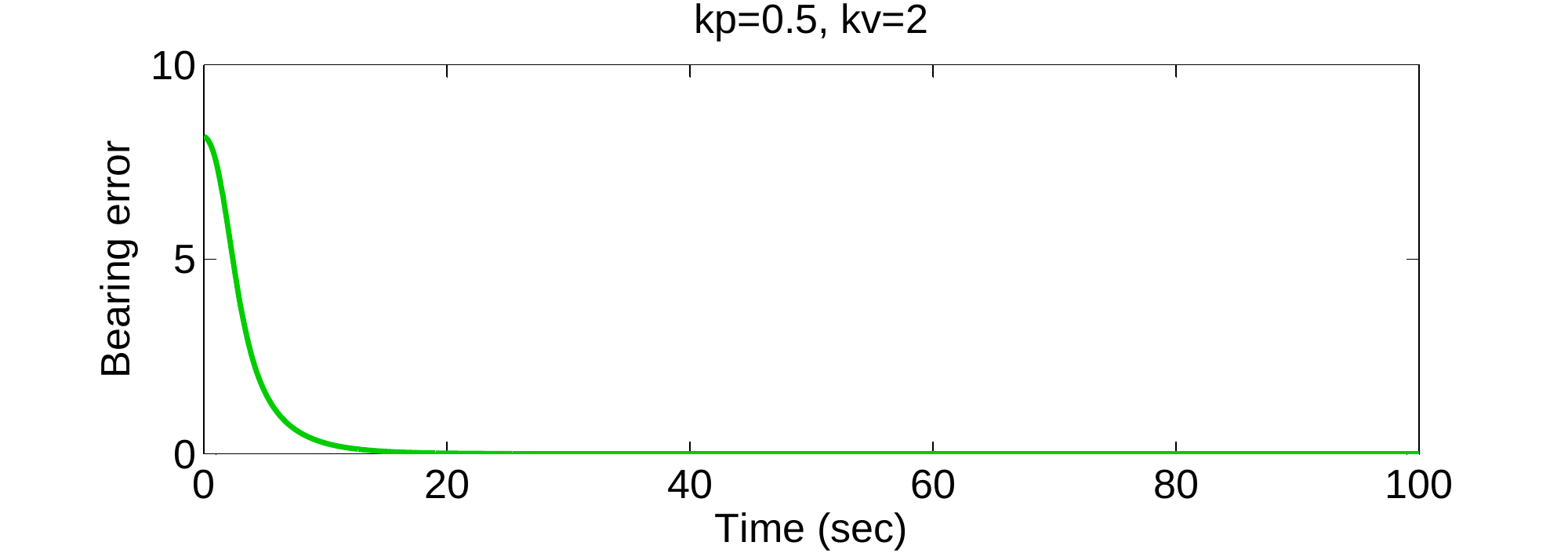}}\\
  \subfloat[Velocity]{\includegraphics[width=\linewidth]{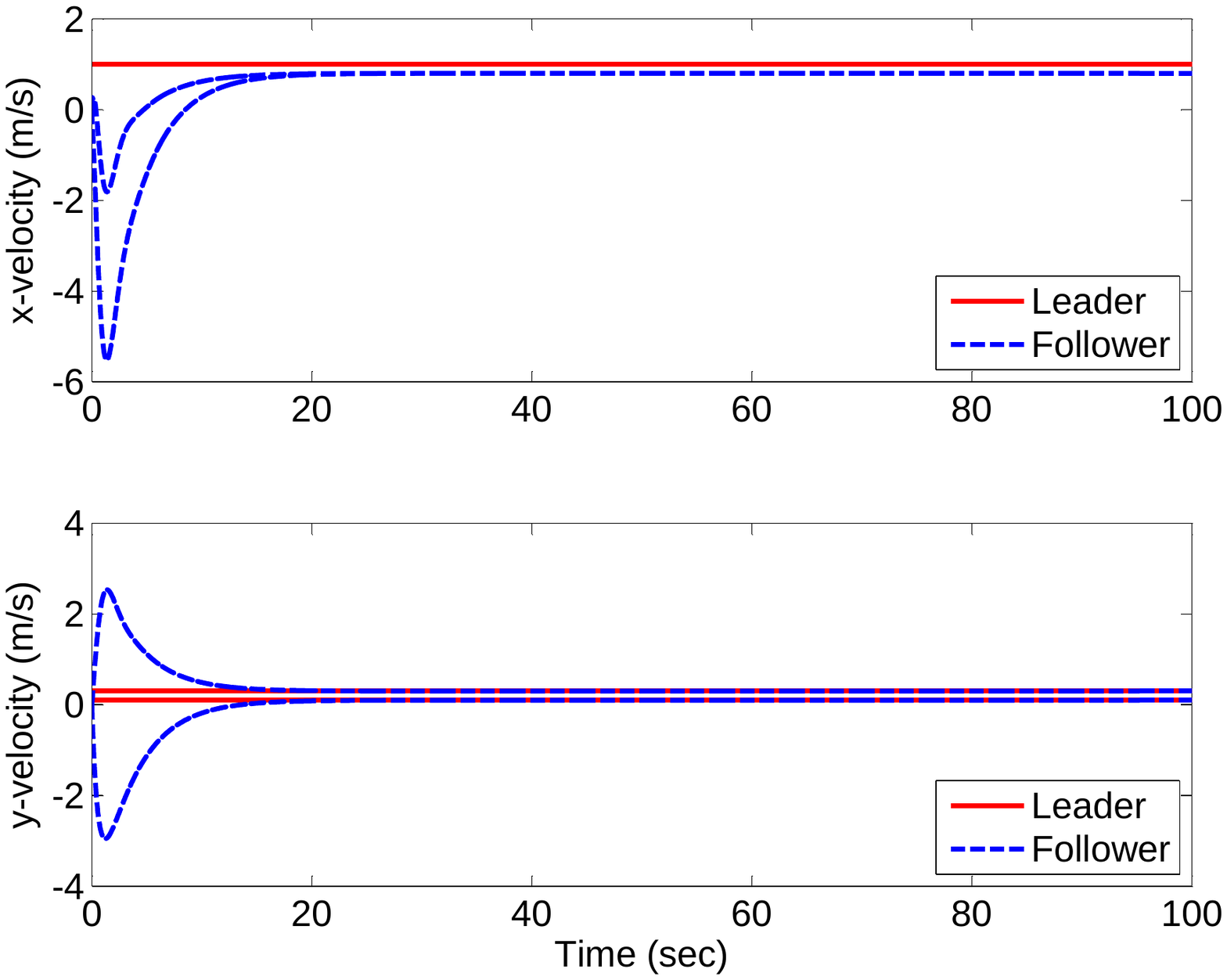}}
  \caption{A simulation example to demonstrate control law \eqref{eq_controlLaw_ConstantVelocity}.}
  \label{fig_sim_2DTransConstant}
\end{figure}

When $v_\ell(t)$ is time-varying (i.e., $\dot{v}_\ell(t)$ is not identically zero), the tracking errors may not converge to zero according to the error dynamics \eqref{eq_errorDynamics_constantVelocity_matrixForm}.
In order to perfectly track target formations with time-varying $v_\ell(t)$, additional acceleration feedback is required as shown in the next subsection.
In practical tasks where the desired target formation has piecewise constant velocities, the control law \eqref{eq_controlLaw_ConstantVelocity} may still give satisfactory performance.

A simulation example is given in Figure~\ref{fig_sim_2DTransConstant} to illustrate control law \eqref{eq_controlLaw_ConstantVelocity}.
The target formation in this example is the square shown in Figure~\ref{fig_Example_Targetformation}(b).
There are two leaders and two followers.
As shown in Figure~\ref{fig_sim_2DTransConstant}(a)--(b), the translation and scale of the formation are continuously varying and, in the meantime, the desired formation pattern is maintained.
In Figure~\ref{fig_sim_2DTransConstant}(c), the x-velocity of each follower converges to a value smaller than that of the leaders, because the velocity of a follower is a combination of the translational and scaling velocities and the scaling velocity in the x-direction is negative in this example.
\subsection{Formation Maneuvering with Time-Varying Leader Velocity}

Now consider the case where $v_\ell(t)$ is time-varying (i.e., $\dot{v}_\ell(t)$ is not identically zero).
Assume $\dot{v}_\ell(t)$ is piecewise continuous.
The following control law handles the time-varying case,
\begin{align}\label{eq_controlLaw_varyingVelocity}
\hspace{-5pt}u_i&=-K_i^{-1}\sum_{j\in\N_i} P_{g^*_{ij}}\left[k_p(p_i-p_j)+k_v(v_i-v_j)-\dot{v}_j\right],
\end{align}
where $K_i=\sum_{j\in\N_i}P_{g_{ij}^*}$.
Compared to control law \eqref{eq_controlLaw_ConstantVelocity}, control law \eqref{eq_controlLaw_varyingVelocity} requires the acceleration of each neighbor.
The design of control law \eqref{eq_controlLaw_varyingVelocity} is inspired by the consensus protocols for tracking time-varying references as proposed in \cite{RenWei2007SCL,RenWei2008TAC}.
The nonsingularity of $K_i$ is guaranteed by the uniqueness of the target formation as shown in the following result.
\begin{lemma}\label{lemma_nonsingularityOfKi}
The matrix $K_i$ is nonsingular for all $i\in\V_f$ if the target formation is unique.
\end{lemma}
\begin{proof}
First of all, the matrix $K_i$ is singular if and only if the bearings $\{g_{ij}^*\}_{j\in\N_i}$ are collinear, because for any $x\in\R^d$, $x^TK_ix=0\Leftrightarrow \sum_{j\in\N_i}x^TP_{g_{ij}^*}x=0\Leftrightarrow P_{g_{ij}^*}x=0, \forall j\in\N_i$.
Since $\Null(P_{g_{ij}^*})=\myspan\{g_{ij}^*\}$, we know $x^TK_ix=0$ if and only if $x$ and $\{g_{ij}^*\}_{j\in\N_i}$ are collinear.
If $\{g_{ij}^*\}_{j\in\N_i}$ are collinear, the follower $p_i^*$ cannot be uniquely determined in the target formation because $p_i^*$ can move along $g_{ij}^*$ without changing any bearings.
As a result, if $K_i$ is singular, the target formation is not unique.
\end{proof}

The convergence of control law \eqref{eq_controlLaw_varyingVelocity} is analyzed below.

\begin{theorem}\label{theorem_convergenceVaryingVelocity}
Under control law \eqref{eq_controlLaw_varyingVelocity}, for any time-varying leader velocity $v_\ell(t)$, the tracking errors $\delta_p(t)$ and $\delta_v(t)$ as defined in \eqref{eq_trackingErrorDef} globally and exponentially converge to zero.
\end{theorem}
\begin{proof}
Multiplying $K_i$ on both sides of control law \eqref{eq_controlLaw_varyingVelocity} gives \begin{align*}
\sum_{j\in\N_i}P_{g_{ij}^*} (\dot{v}_i-\dot{v}_j)&=\sum_{j\in\N_i} P_{g^*_{ij}}\left[-k_p(p_i-p_j)-k_v(v_i-v_j)\right],
\end{align*}
whose matrix-vector form is
\begin{align*}
\L_{ff}\dot{v}_f+\L_{f\ell}\dot{v}_\ell
&=-k_p(\L_{ff}p_f+\L_{f\ell}p_\ell)-k_v(\L_{ff}v_f+\L_{f\ell}v_\ell) \nonumber\\
&=-k_p\L_{ff}\delta_p-k_v\L_{ff}\delta_v.
\end{align*}
It follows that $\dot{v}_f=-k_p\delta_p-k_v\delta_v-\L_{ff}^{-1}\L_{f\ell}\dot{v}_\ell$.
Then the tracking error dynamics are $\dot{\delta}_p=\delta_v$ and $\dot{\delta}_v=\dot{v}_f+\L_{ff}^{-1}\L_{f\ell}\dot{v}_\ell=-k_p\delta_p+k_v\delta_v$,
which are expressed in a compact form as
\begin{align}\label{eq_errorDynamics_varyingVelocity_matrixForm}
\left[
  \begin{array}{c}
    \dot{\delta}_p \\
    \dot{\delta}_v \\
  \end{array}
\right]
=
\left[
  \begin{array}{cc}
    0 & I \\
    -k_pI & -k_vI \\
  \end{array}
\right]
\left[
  \begin{array}{c}
    \delta_p \\
    \delta_v \\
  \end{array}
\right].
\end{align}
The eigenvalue of the state matrix is $\lambda=(-k_v\pm\sqrt{k_v^2-4k_p})/2$, which is always in the open left-half plane for any $k_p, k_v>0$.
The global and exponential convergence result follows.
\end{proof}

\begin{figure}
  \centering
  \subfloat[Trajectory (The dark area stands for an obstacle.)]{\includegraphics[width=\linewidth]{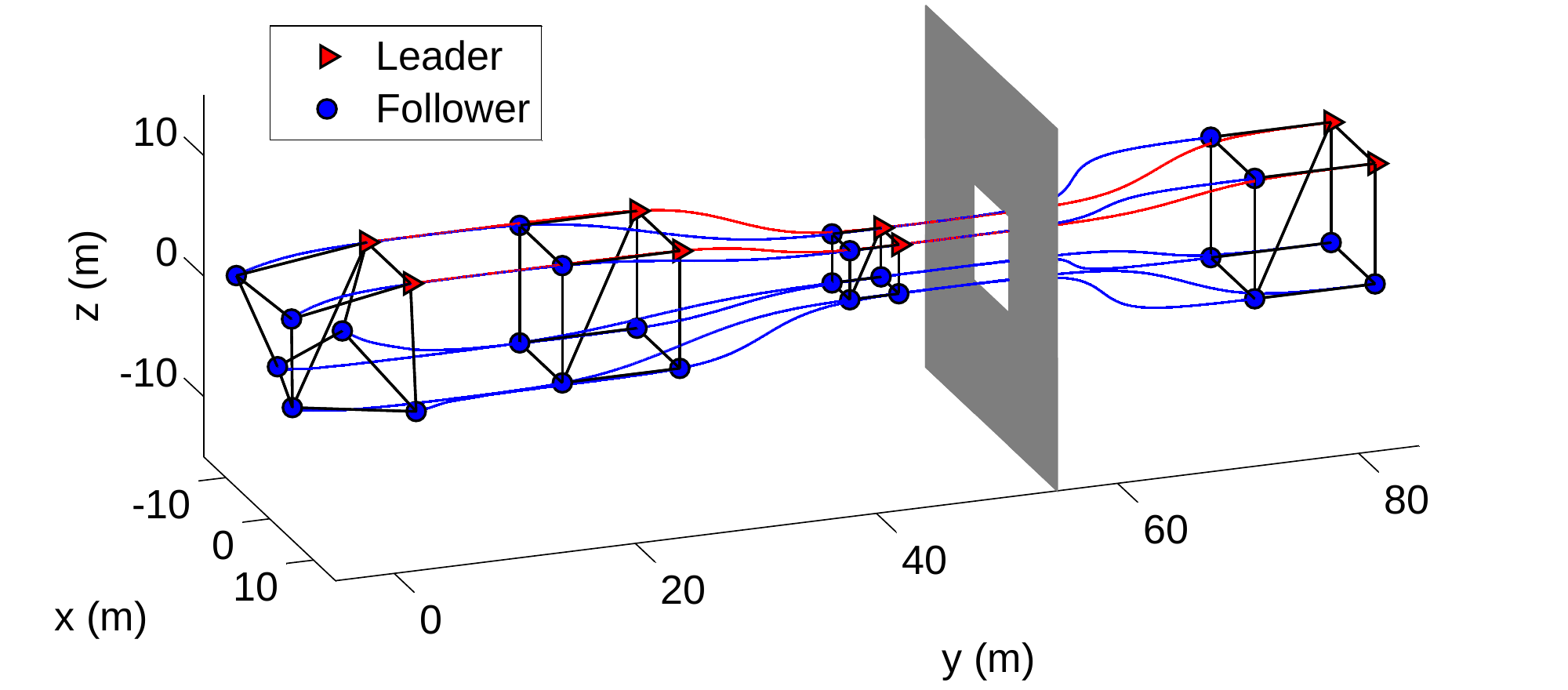}}\\
  \subfloat[Total bearing error: $\sum_{(i,j)\in\E}\|g_{ij}(t)-g_{ij}^*\|$]{\includegraphics[width=\linewidth]{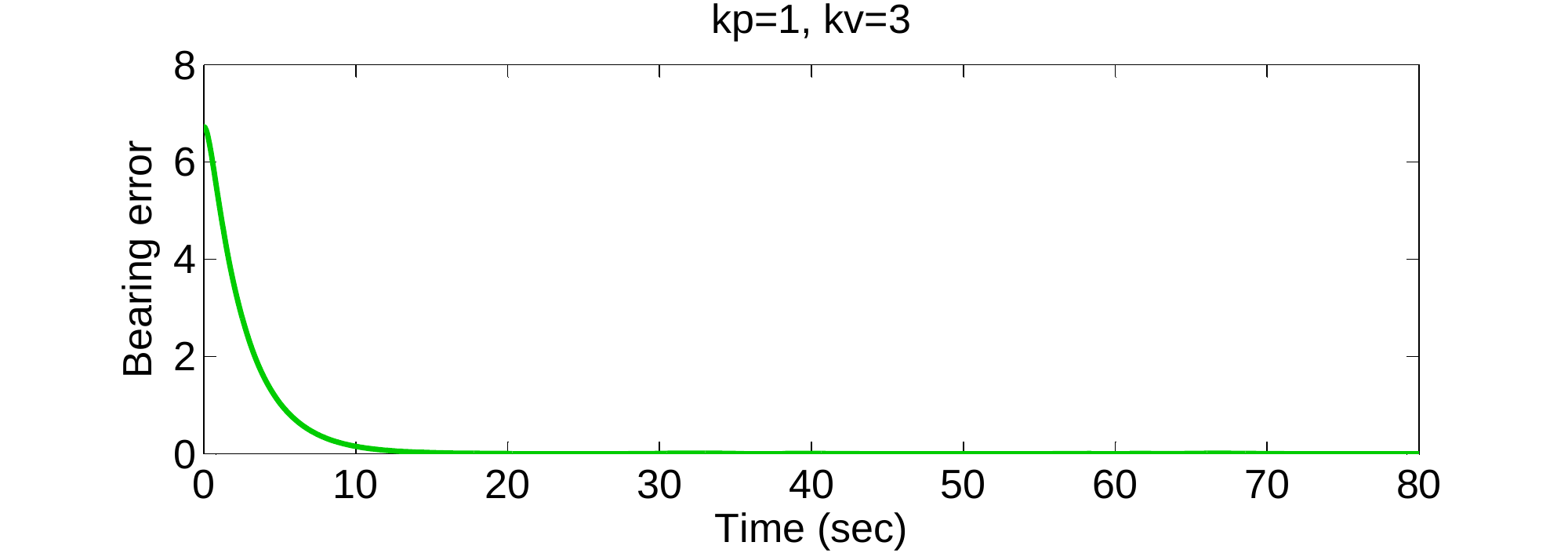}}\\
  \subfloat[Velocity]{\includegraphics[width=\linewidth]{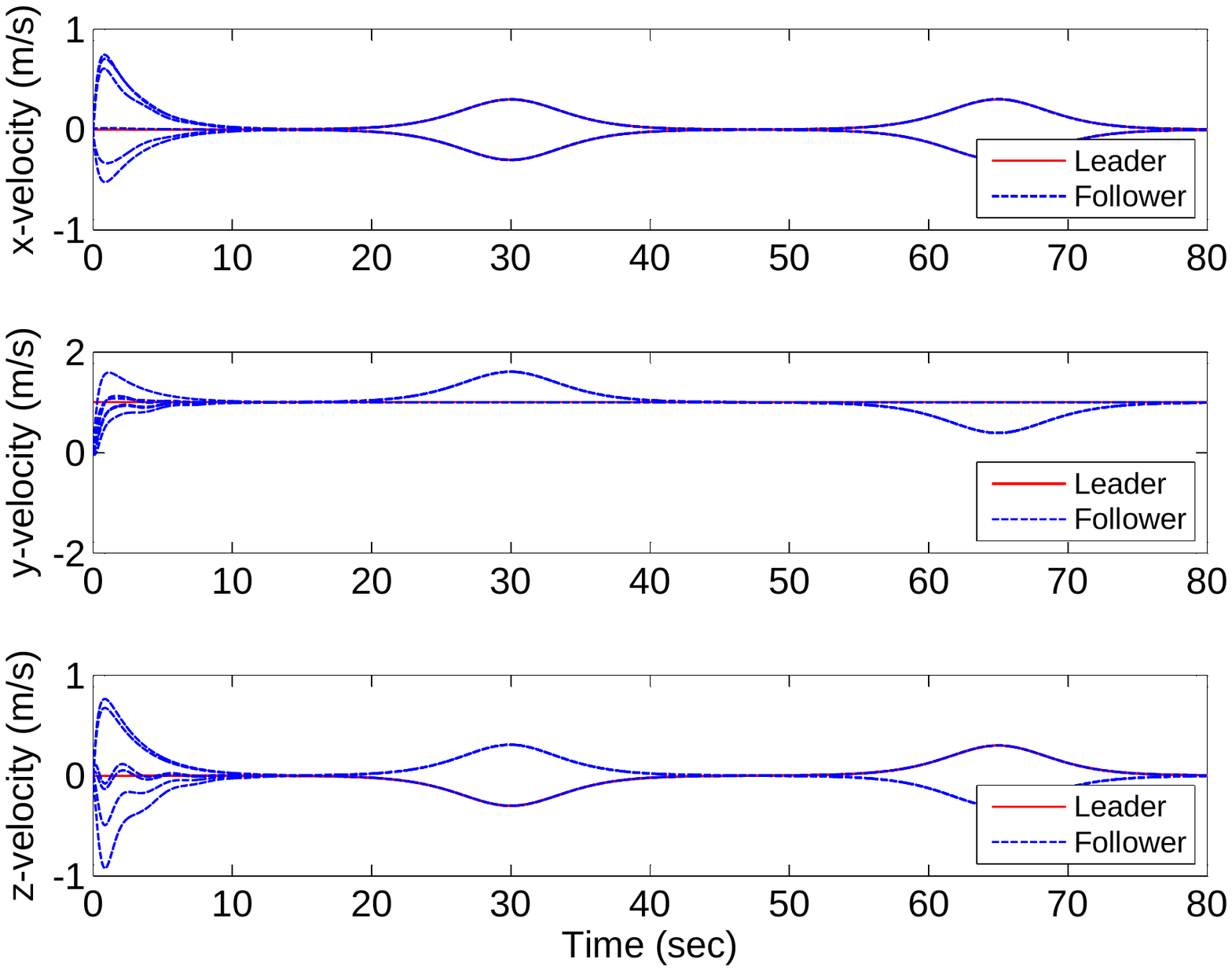}}
  \caption{A simulation example to demonstrate control law \eqref{eq_controlLaw_varyingVelocity}.}
  \label{fig_sim_3DTransVaryingVelocity}
\end{figure}

By comparing the error dynamics in \eqref{eq_errorDynamics_varyingVelocity_matrixForm} and \eqref{eq_errorDynamics_constantVelocity_matrixForm}, we see that the role of the acceleration feedback in control law \eqref{eq_controlLaw_varyingVelocity} is to eliminate the term that contains $\dot{v}_\ell(t)$ so that it does not affect the convergence of the errors.

A simulation example is shown in Figure~\ref{fig_sim_3DTransVaryingVelocity} to illustrate control law \eqref{eq_controlLaw_varyingVelocity}.
The target formation in this example is the three-dimensional cube shown in Figure~\ref{fig_Example_Targetformation}(c), which has two leaders and six followers.
As shown in Figure~\ref{fig_sim_3DTransVaryingVelocity}(a)--(b), the translation and scale of the formation are continuously varying and, in the meantime, the formation converges from an initial configuration to the desired pattern.
Although the velocities of the leaders are time-varying, the desired formation pattern is maintained exactly during the formation evolution.

The simulation example also demonstrates that the proposed control law can be used for obstacle avoidance, such as passing through narrow passages.
In practice, collision avoidance requires sophisticated mechanisms such as obstacle detection and path generation (see, for example, \cite{das2002vision}).
Details on obstacle avoidance are out of the scope of this paper.

\section{Bearing-Based Formation Control with Practical Issues}\label{section_practicalIssues}
In this section, we consider bearing-based formation control in the presence of some issues that may appear in practical implementations, including input disturbances, input saturation, and collision avoidance among the agents.

\subsection{Constant Input Disturbance}

Suppose there exists an unknown constant input disturbance for each follower.
The dynamics of follower $i\in\V_f$ are
\begin{align*}
\dot{p}_i= v_i, \quad \dot{v}_i=u_i+\mathrm{w}_i,
\end{align*}
where $\mathrm{w}_i\in\R^d$ is an unknown constant signal, and let $\mathrm{w}_f=[\mathrm{w}_1^T,\dots,\mathrm{w}_{n_f}^T]^T$.
In practice, the constant input disturbance might be caused by, for example, constant sensor or actuator biases.
In order to handle the input disturbance, we add an integral control term to control law \eqref{eq_controlLaw_ConstantVelocity} and obtain
\begin{align}\label{eq_controlLaw_integral_ConstantVelocity}
u_i&=-\sum_{j\in\N_i}P_{g_{ij}^*}\bigg[k_p(p_i-p_j)+k_v(v_i-v_j)  \bigg. \nonumber\\
&\qquad\qquad\qquad\quad \bigg. +k_I\int_0^t (p_i-p_j)\D \tau\bigg],
\end{align}
where $k_I>0$ is the constant integral control gain.
We next show that the integral control will not only eliminate the impact of the constant
disturbance but will also handle the case where $\dot{v}_\ell(t)$ is nonzero and constant.

\begin{theorem}\label{theorem_convergenceIntegralConstant}
Consider the control law \eqref{eq_controlLaw_integral_ConstantVelocity} with constant disturbance $\mathrm{w}_f$ and constant leader acceleration $\dot{v}_\ell$. If the control gains satisfy $0<k_I<k_pk_v\lambda_{\min}(\L_{ff})$, then the tracking errors $\delta_p(t)$ and $\delta_v(t)$ globally and exponentially converge to zero.
\end{theorem}
\begin{proof}
The matrix-vector form of control law~\eqref{eq_controlLaw_integral_ConstantVelocity} is
\begin{align*}
\dot{v}_f&=-k_I\int_0^t(\L_{ff}p_f+\L_{f\ell}p_\ell)\D \tau-k_p(\L_{ff}p_f+\L_{f\ell}p_\ell)\\
&\qquad\qquad-k_v(\L_{ff}p_f+\L_{f\ell}p_\ell)+\mathrm{w}_f \\
&=-k_I\L_{ff}\int_0^t \delta_p\D \tau -k_p\L_{ff}\delta_p-k_v\L_{ff}\delta_v+\mathrm{w}_f.
\end{align*}
Denote $\eta\triangleq\int_0^t \delta_p\D \tau$. It then follows that $\dot{\eta}=\delta_p$, $\dot{\delta}_p=\delta_v$, and $\dot{\delta}_v=\dot{v}_f+\L_{ff}^{-1}\L_{f\ell}\dot{v}_\ell=-k_I\L_{ff}\eta -k_p\L_{ff}\delta_p-k_v\L_{ff}\delta_v+\mathrm{w}_f+\L_{ff}^{-1}\L_{f\ell}\dot{v}_\ell$, the matrix-vector form of which is given by
\begin{align}\label{eq_errorDynamics_integral_constantVelocity_matrixForm}
\left[
  \begin{array}{c}
    \dot{\eta} \\
    \dot{\delta}_p \\
    \dot{\delta}_v \\
  \end{array}
\right]
&=
\left[
  \begin{array}{ccc}
    0 & I & 0 \\
    0 & 0 & I \\
    -k_I\L_{ff} & -k_p\L_{ff} & -k_v\L_{ff} \\
  \end{array}
\right]
\left[
  \begin{array}{c}
    \eta \\
    \delta_p \\
    \delta_v \\
  \end{array}
\right]\nonumber\\
&\qquad+\left[
  \begin{array}{c}
    0 \\
    0 \\
    \mathrm{w}_f \\
  \end{array}
\right]
+\left[
  \begin{array}{c}
    0 \\
    0 \\
    \L_{ff}^{-1}\L_{f\ell} \\
  \end{array}
\right]\dot{v}_\ell.
\end{align}
Denote $A$ as the state matrix of the above dynamics with $\lambda$ an associated eigenvalue. We next identify the condition for $\Re(\lambda)<0$.
Note the state matrix is in the controllable canonical form.
Then the characteristic polynomial is
\begin{align*}
\det(\lambda I-A)=\det(\lambda^3I+k_v\L_{ff}\lambda^2+k_p\L_{ff}\lambda+k_I\L_{ff}).
\end{align*}
As a result, $\lambda^3$ can be viewed as an eigenvalue of the matrix $-(k_v\lambda^2+k_v\lambda+k_I)\L_{ff}$.
By denoting $\mu$ as an eigenvalue of $\L_{ff}$, we have $\lambda^3+k_v\mu\lambda^2+k_p\mu\lambda+k_I\mu=0$.
By the Routh-Hurwitz stability criterion, we have $\Re(\lambda)<0$ if and only $k_v\mu, k_p\mu, k_I\mu>0$ and $(k_v\mu)(k_p\mu)>k_I\mu$.
Since $k_v, k_p, \mu>0$, we have $0<k_I<k_vk_p\mu$.
In order to make $\Re(\lambda)<0$ for all $\mu$, it is required $0<k_I<k_vk_p\lambda_{\min}(\L_{ff})$ where $\lambda_{\min}(\L_{ff})$ is the minimum eigenvalue of $\L_{ff}$.
When $A$ is Hurwitz, given constant $\mathrm{w}_f$ and $\dot{v}_\ell$, the steady state is $\delta_p(\infty)=\delta_v(\infty)=0$ and $\eta(\infty)=-\L_{ff}^{-1}(\mathrm{w}_f+\L_{ff}^{-1}\L_{f\ell}\dot{v}_\ell)/k_I$.
\end{proof}

As can be seen from the error dynamics \eqref{eq_errorDynamics_integral_constantVelocity_matrixForm}, when $\dot{v}_\ell$ is constant, it has the same impact as an input disturbance and hence is handled by the integral control.
The idea of integral control has also been applied in consensus, distance-based and bearing-based formation maneuver control problems \cite{Andreasson2014TAC,oshri2015ECC,zhao2015MSC}.
It is also interesting to note that the integral control gain must be bounded by $\lambda_{\min}(\L_{ff})$, which we expect should have graph-theoretic interpretations and is the subject of future work.

Similarly, by adding an integral control term to control law \eqref{eq_controlLaw_varyingVelocity}, we obtain the following control law that can handle the unknown constant input disturbance and time-varying $v_\ell(t)$,
\begin{align}\label{eq_controlLaw_integral_varyingVelocity}
u_i&=-K_i^{-1}\sum_{j\in\N_i}P_{g_{ij}^*}\bigg[k_p(p_i-p_j)+k_v(v_i-v_j)-\dot{v}_j\bigg. \nonumber\\
&\qquad\qquad\qquad\qquad\quad \bigg.+k_I\int_0^t (p_i-p_j)\D \tau\bigg].
\end{align}
The convergence result for control law \eqref{eq_controlLaw_integral_varyingVelocity} is given below. The proof is similar to Theorem~\ref{theorem_convergenceIntegralConstant} and omitted.

\begin{theorem}
Consider the control law~\eqref{eq_controlLaw_integral_varyingVelocity} with constant disturbance $\mathrm{w}_f$ and  time-varying leader velocity ${v}_\ell(t)$.  If the control gains satisfy $0<k_I<k_pk_v$, then the tracking errors $\delta_p(t)$ and $\delta_v(t)$ globally and exponentially converge to zero.
\end{theorem}

\subsection{Acceleration Saturation}

In practical implementations, the acceleration input is always bounded.
In the presence of acceleration saturation, the control law \eqref{eq_controlLaw_ConstantVelocity} becomes
\begin{align}\label{eq_controlLaw_saturation_constantVelocity}
u_i&=\sat\left\{-\sum_{j\in\N_i}P_{g_{ij}^*}\left[k_p(p_i-p_j)+k_v(v_i-v_j)\right]\right\},
\end{align}
where $\sat(\cdot)$ is a saturation function that is either $\sat(x)=\sign(x)\min\{|x|,\beta\}$ or $\sat(x)=\beta\tanh(x)$ where $x\in\R$ and $\beta>0$ is the constant bound for $|x|$.
For a vector $x=[x_1,\dots,x_q]^T\in\R^q$, $\sat(x)$ is defined component-wise as $\sat(x)=[\sat(x_1),\dots,\sat(x_q)]^T$.

Due to the saturation function, the formation dynamics become nonlinear and the formation stability can be proven by a Lyapunov approach.
Inspired by the work in \cite{MengZiyang2013SCL}, we introduce the integral function $\Phi(x)\triangleq\int_0^x \sat(\tau)\D \tau$ for $x\in\R$.
Due to the properties of $\sat(\cdot)$, we have that $\Phi(x)\ge0$ for all $x\in\R$ and $\Phi(x)=0$ if and only if $x=0$.
In the case of $\sat(x)=\beta\tanh(x)$, we have $\Phi(x)=\beta\log(\cosh(x))$.
For a vector $x=[x_1,\dots,x_q]^T\in\R^q$, $\Phi(x)$ is defined component-wise as
\begin{align*}
\Phi(x)=\left[\int_0^{x_1} \sat(\tau)\D \tau, \dots, \int_0^{x_q} \sat(\tau)\D \tau\right]^T\in\R^q.
\end{align*}
The useful properties of $\Phi(\cdot)$ and $\sat(\cdot)$ are given below.
\begin{lemma}\label{lemma_properties_PhiFunction}
Given $x(t)\in\R^q$, the quantity $\one^T\Phi(x)$ satisfies
\begin{enumerate}[(a)]
\item $\one^T\Phi(x)\ge0$ and $\one^T\Phi(x)=0$ if and only if $x=0$.
\item ${\D (\one^T\Phi(x))}/{\D t}=\dot{x}^T\sat(x)$.
\end{enumerate}
\end{lemma}
\begin{proof}
Note $\one^T\Phi(x)=\sum_{i=1}^q \Phi(x_i)$.
Since $\Phi(x_i)\ge0$ and $\Phi(x_i)=0$ if and only if $x_i=0$, property (a) is proven.
The time derivative of $\one^T\Phi(x)$ is given by $\D(\one^T\Phi(x))/\D t=\sum_{i=1}^q \dot{\Phi}(x_i)=\sum_{i=1}^q \dot{x}_i\sat(x_i)=\dot{x}^T\sat(x)$.
\end{proof}
\begin{lemma}\label{lemma_satFunction_property}
For any two vectors $x, y\in\R^q$, it always holds that $y^T\left[\sat(x-y)-\sat(x)\right]\le0$ and $y^T\left[\sat(x)-\sat(x+y)\right]\le0$.
Moveover, if $\sat(\cdot)$ is strictly monotonic, the equalities hold if and only if $y=0$.
\end{lemma}
\begin{proof}
We only prove the first inequality; the second one can be proven similarly.
Note $y^T\left(\sat(x-y)-\sat(x)\right)=\sum_{i=1}^q y_i(\sat(x_i-y_i)-\sat(x_i))$.
It follows from the monotonicity of the saturation function that $\sat(x_i-y_i)-\sat(x_i)\ge0$ if $y_i<0$, and $\sat(x_i-y_i)-\sat(x_i)\le0$ if $y_i>0$.
Therefore, $y_i(\sat(x_i-y_i)-\sat(x_i))\le0$ for all $x_i, y_i\in\R$.
If $\sat(\cdot)$ is strictly monotonic, $y_i(\sat(x_i-y_i)-\sat(x_i))=0$ if and only if $y_i=0$, which completes the proof.
\end{proof}

With the above preparation, we now analyze the formation stability under control law \eqref{eq_controlLaw_saturation_constantVelocity}.

\begin{theorem}\label{theorem_ControlSaturationConstant}
Under control law \eqref{eq_controlLaw_saturation_constantVelocity} with a constant leader velocity $v_\ell(t)$, the tracking errors $\delta_p(t)$ and $\delta_v(t)$ globally and asymptotically converge to zero.
\end{theorem}
\begin{proof}
The matrix-vector form of control law \eqref{eq_controlLaw_saturation_constantVelocity} is $\dot{v}_f=\sat(-k_p\L_{ff}\delta_p-k_v\L_{ff}\delta_v)$.
Substituting $\dot{v}_f$ and $\dot{v}_\ell=0$ into the tracking error dynamics gives $\dot{\delta}_v=\dot{v}_f+\L_{ff}^{-1}\L_{f\ell}\dot{v}_\ell=\sat(-k_p\L_{ff}\delta_p-k_v\L_{ff}\delta_v)$.
Consider the Lyapunov function
\begin{align*}
V&=\one^T\Phi(-k_p\L_{ff}\delta_p-k_v\L_{ff}\delta_v) + \one^T\Phi(-k_p\L_{ff}\delta_p)\\
&\qquad+k_p\delta_v^T\L_{ff}\delta_v.
\end{align*}
It follows from Lemma~\ref{lemma_properties_PhiFunction}(a) that $V\ge0$ and $V=0$ if and only if $\delta_p=\delta_v=0$.
According to Lemma~\ref{lemma_properties_PhiFunction}(b), the time derivative of the Lyapunov function is
\begin{align*}
\dot{V}
&=(-k_p\L_{ff}\delta_v-k_v\L_{ff}\dot{\delta}_v)^T\sat(-k_p\L_{ff}\delta_p-k_v\L_{ff}\delta_v) \\
&\quad+ (-k_p\L_{ff}\delta_v)^T\sat(-k_p\L_{ff}\delta_p)+2k_p\delta_v^T\L_{ff}\dot{\delta}_v.
\end{align*}
It follows from $\sat(-k_p\L_{ff}\delta_p-k_v\L_{ff}\delta_v)=\dot{\delta}_v$ that
\begin{align*}
\dot{V}
&=-(k_p\L_{ff}\delta_v)^T\dot{\delta}_v -(k_v\L_{ff}\dot{\delta}_v)^T\dot{\delta}_v \\
&\quad+ (-k_p\L_{ff}\delta_v)^T\sat(-k_p\L_{ff}\delta_p)+2k_p\delta_v^T\L_{ff}\dot{\delta}_v\\
&=-k_v\dot{\delta}_v^T\L_{ff}\dot{\delta}_v -(k_p\L_{ff}\delta_v)^T\sat(-k_p\L_{ff}\delta_p)+k_p\delta_v^T\L_{ff}\dot{\delta}_v\\
&=-k_v\dot{\delta}_v^T\L_{ff}\dot{\delta}_v \\ &\,+k_p\delta_v^T\L_{ff}[\sat(-k_p\L_{ff}\delta_p-k_v\L_{ff}\delta_v)-\sat(-k_p\L_{ff}\delta_p)],
\end{align*}
where the first term $-k_v\dot{\delta}_v^T\L_{ff}\dot{\delta}_v$ is nonpositive and the second term is also nonpositive according to Lemma~\ref{lemma_satFunction_property}.
As a result, $\dot{V}\le 0$ for all $t\ge0$.

We next identify the invariant set for $\dot{V}=0$.
When $\dot{V}=0$, we have
\begin{align}
-k_v\dot{\delta}_v^T\L_{ff}\dot{\delta}_v&=0, \label{eq_dotV=0_term1}\\
k_p\delta_v^T\L_{ff}[\dot{\delta}_v-\sat(-k_p\L_{ff}\delta_p)]&=0. \label{eq_dotV=0_term2}
\end{align}
It follows from \eqref{eq_dotV=0_term1} that $\dot{\delta}_v=\sat(-k_p\L_{ff}\delta_p-k_v\L_{ff}\delta_v)=0$, which further implies $k_p\L_{ff}\delta_p=-k_v\L_{ff}\delta_v$.
It then follows from \eqref{eq_dotV=0_term2} that $k_p\delta_v^T\L_{ff}\sat(-k_v\L_{ff}\delta_v)=0$, which indicates $\L_{ff}\delta_v=0\Leftrightarrow \delta_v=0$ because $\L_{ff}\delta_v$ and $\sat(\L_{ff}\delta_v)$ have the same sign componentwise.
Since $k_p\L_{ff}\delta_p=-k_v\L_{ff}\delta_v$, we have $\delta_p=0$.
Therefore, $\dot{V}=0$ if and only if $\delta_p=\delta_v=0$. According to the invariance principle, the tracking errors $\delta_p$ and $\delta_v$ globally and asymptotically converges to zero.
\end{proof}

In order to handle input saturation in the case of time-varying ${v}_\ell(t)$, we use the control law
\begin{align}\label{eq_controlLaw_saturation_varyingVelocity}
u_i&=K_i^{-1}\sat\left\{-\sum_{j\in\N_i}P_{g_{ij}^*}\left[k_p(p_i-p_j)+k_v(v_i-v_j)\right]\right\}\nonumber\\
&\qquad+K_i^{-1}\sum_{j\in\N_i}P_{g_{ij}^*}\dot{v}_j,
\end{align}
where $K_i=\sum_{j\in\N_i}P_{g_{ij}^*}$.
Although the saturation function is not applied to the entire acceleration input, the above control law ensures bounded input given arbitrary initial conditions.
In particular, under control law \eqref{eq_controlLaw_saturation_varyingVelocity} the velocity dynamics are
$\sum_{j\in\N_i}P_{g_{ij}^*}(\dot{v}_i-\dot{v}_j)=\sat(\star)$, where the quantity in the saturation function is written in short as $\sat(\star)$.
Then the matrix-vector form of the velocity dynamics is $\L_{ff}\dot{v}_f+\L_{f\ell}\dot{v}_\ell=\sat(\star)$, which implies $\dot{v}_f=\L_{ff}^{-1}\sat(\star)-\L_{ff}^{-1}\L_{f\ell}\dot{v}_\ell$.
It follows that
\begin{align*}
\|\dot{v}_f\|_\infty\le\|\L_{ff}^{-1}\|_\infty\|\sat(\star)\|_\infty+\|\L_{ff}^{-1}\L_{f\ell}\|_\infty\|\dot{v}_\ell\|_\infty.
\end{align*}
The upper bound for the acceleration as shown above is independent to the initial conditions of the formation position or velocity.
It relies on the rigidity structure of the target formation and the magnitude of the accelerations of the leaders.
We next characterize the global formation stability under control law \eqref{eq_controlLaw_saturation_varyingVelocity}.

\begin{theorem}\label{theorem_saturationControlVaryingVelocity}
Under control law \eqref{eq_controlLaw_saturation_varyingVelocity} and for any time-varying leader velocity $v_\ell(t)$, the tracking errors $\delta_p(t)$ and $\delta_v(t)$ globally and asymptotically converge to zero.
\end{theorem}
\begin{proof}
Let $\varepsilon_i\triangleq\sum_{j\in\N_i}P_{g_{ij}^*}(p_i-p_j)$.
It follows from \eqref{eq_controlLaw_saturation_varyingVelocity} that $\dot{\varepsilon}_i=\sat(-k_p\varepsilon_i-k_v\dot{\varepsilon}_i)$.
By denoting $\varepsilon=[\varepsilon_1,\dots,\varepsilon_{n_f}]^T$, we obtain
\begin{align*}
    \dot{\varepsilon}=\sat(-k_p\varepsilon-k_v\dot{\varepsilon}).
\end{align*}
Note $\varepsilon=\L_{ff}p_f+\L_{f\ell}p_\ell=\L_{ff}\delta_p$ and hence $\dot{\varepsilon}=\L_{ff}\delta_v$.
As a result, $\varepsilon=\dot{\varepsilon}=0\Leftrightarrow \delta_p=\delta_v=0$.
We prove $\delta_p, \delta_v\rightarrow0$ by showing $\varepsilon,\dot{\varepsilon}\rightarrow0$.
To that end, consider the Lyapunov function
\begin{align*}
V=\one^T\Phi(-k_p\varepsilon-k_v\dot{\varepsilon})+\one^T\Phi(-k_p\varepsilon)+k_p\dot{\varepsilon}^T\dot{\varepsilon}.
\end{align*}
The time derivative of $V$ is given by
\begin{align*}
\dot{V}
&=(-k_p\dot{\varepsilon}-k_v\ddot{\varepsilon})\sat(-k_p\varepsilon-k_v\dot{\varepsilon})\\
&\qquad+(-k_p\dot{\varepsilon})\sat(-k_p\varepsilon)+2k_p\dot{\varepsilon}^T\ddot{\varepsilon}\\
&=-k_v\ddot{\varepsilon}^T\ddot{\varepsilon}+k_p\dot{\varepsilon}^T[\sat(k_p\varepsilon)-sat(k_p\varepsilon+k_v\dot{\varepsilon})].
\end{align*}
Similar to the proof of Theorem~\ref{theorem_ControlSaturationConstant}, it can be shown that $\dot{V}\le0$ and the invariant set where $\dot{V}=0$ is $\varepsilon=\dot{\varepsilon}=0$.
Therefore, by the invariance principle, $\varepsilon$ and $\dot{\varepsilon}$ globally and asymptotically converge to zero, and so do $\delta_p$ and $\delta_v$.
\end{proof}

\subsection{A Collision-Free Condition}

Collision avoidance among the agents is an important issue in practical formation control problems.
The proposed control laws can be implemented together with, for example, artificial potentials \cite{Han2015MSC} to ensure collision avoidance.
In this work, we propose a sufficient condition on the initial formation that ensures no collision between any pair of agents (even they are not neighbors).
Suppose $\gamma$ is the desired minimum distance that should be guaranteed between any two agents and $\gamma$ satisfies
\begin{align*}
0\le \gamma<\min_{i,j\in\V, t\ge0}\|p_i^*(t)-p_j^*(t)\|.
\end{align*}

\begin{theorem}\label{theorem_collisionAvoidanceCondForDelta}
Under control law \eqref{eq_controlLaw_ConstantVelocity}, for any constant leader velocity $v_\ell$, it is guaranteed that
\begin{align*}
\|p_i(t)-p_j(t)\|>\gamma, \quad \forall i,j\in\V, \forall t\ge0,
\end{align*}
if $\delta_p(0)$ and $\delta_v(0)$ satisfy
\begin{align}\label{eq_collisionAvoidanceCondition_constantVelocity}
&{k_p\delta_p^T(0)\L_{ff}\delta_p(0)+\delta_v^T(0)\delta_v(0)} \nonumber\\
&\quad<{\frac{k_p\lambda_{\min}(\L_{ff})}{n_f}}\left(\min_{i,j\in\V,t\ge0}\|p_i^*(t)-p_j^*(t)\|-\gamma\right)^2.
\end{align}
\end{theorem}
\begin{proof}
For any $i,j\in\V$, it always holds that
\begin{align*}
&p_i(t)-p_j(t)\\
&\quad\equiv[p_i^*(t)-p_j^*(t)]+[p_i(t)-p_i^*(t)]-[p_j(t)-p_j^*(t)],
\end{align*}
where $p_i^*(t)$ and $p_j^*(t)$ are the expected positions for agents $i$ and $j$ in the target formation.
Note $p_i(t)-p_i^*(t)\equiv0$ for $i\in\V_\ell$.
It follows that
\begin{align}\label{eq_interDistanceInequality}
    &\|p_i(t)-p_j(t)\|\nonumber\\
    &\quad\ge \|p_i^*(t)-p_j^*(t)\|-\|p_i(t)-p_i^*(t)\|-\|p_j(t)-p_j^*(t)\| \nonumber\\
    &\quad=\|p_i^*(t)-p_j^*(t)\|-\sum_{k\in\V_f} \|p_k(t)-p_k^*(t)\|\nonumber\\
    &\quad\ge\|p_i^*(t)-p_j^*(t)\|-\sqrt{n_f}\|p_f(t)-p_f^*(t)\| \nonumber\\
    &\quad=\|p_i^*(t)-p_j^*(t)\|-\sqrt{n_f}\|\delta_p(t)\|,    \quad \forall t\ge0.
\end{align}
The above inequality gives a lower bound for $\|p_i(t)-p_j(t)\|$. If we can find a condition such that the lower bound is always greater than $\gamma$, then the minimum distance $\gamma$ can be guaranteed.
In this direction, consider the Lyapunov function
\begin{align*}
V(\delta_p(t),\delta_v(t))=k_p\delta_p^T(t)\L_{ff}\delta_p(t)+\delta_v^T(t)\delta_v(t).
\end{align*}
With the error dynamics as given in \eqref{eq_errorDynamics_constantVelocity_matrixForm}, the time derivative of $V$ along the error dynamics is $\dot{V}=-2k_v\delta_v^T\L_{ff}\delta_v\le0$.
As a result, we have
\begin{align*}
k_p\lambda_{\min} (\L_{ff})\|\delta_p(t)\|^2
&\le k_p\delta_p^T(t)\L_{ff}\delta_p(t) \\
&\le k_p\delta_p^T(t)\L_{ff}\delta_p(t)+\delta_v^T(t)\delta_v(t) \\
&\le V(\delta_p(0),\delta_v(0)),
\end{align*}
which implies
\begin{align}\label{eq_deltapnormbound}
\|\delta_p(t)\|\le \sqrt{\frac{V(\delta_p(0),\delta_v(0))}{k_p\lambda_{\min}(\L_{ff})}}.
\end{align}
By combining \eqref{eq_interDistanceInequality} and \eqref{eq_deltapnormbound}, we have that $\|p_i(t)-p_j(t)\|>\gamma$ for all $t\ge0$ and all $i,j\in\V$ if $\delta_p(0)$ and $\delta_v(0)$ satisfies
\begin{align*}
\min_{i,j\in\V, t\ge0}\|p_i^*(t)-p_j^*(t)\|-\sqrt{\frac{n_fV(\delta_p(0),\delta_v(0))}{k_p\lambda_{\min}(\L_{ff})}} >\gamma,
\end{align*}
which can be rewritten as \eqref{eq_collisionAvoidanceCondition_constantVelocity}.
\end{proof}

The intuition behind the condition in Theorem~\ref{theorem_collisionAvoidanceCondForDelta} is that collision avoidance is guaranteed if the initial formation is sufficiently close to the target formation.
Theorem~\ref{theorem_collisionAvoidanceCondForDelta} is merely applicable in the case of constant $v_\ell(t)$.
For time-varying $v_\ell(t)$, we have a similar condition for control law \eqref{eq_controlLaw_varyingVelocity}.
The proof is similar to Theorem~\ref{theorem_collisionAvoidanceCondForDelta} and omitted.

\begin{theorem}\label{theorem_collisionAvoidanceCondForDelta_withAcc}
Under control law \eqref{eq_controlLaw_varyingVelocity}, for any time-varying leader velocity $v_\ell(t)$, it can be guaranteed that
\begin{align*}
\|p_i(t)-p_j(t)\|>\gamma, \quad \forall i,j\in\V, \forall t\ge0,
\end{align*}
if $\delta_p(0)$ and $\delta_v(0)$ satisfy
\begin{align*}
&{k_p\delta_p^T(0)\delta_p(0)+\delta_v^T(0)\delta_v(0)}\nonumber \\
&\quad <{\frac{k_p}{n_f}}\left(\min_{i,j\in\V,t\ge0}\|p_i^*(t)-p_j^*(t)\|-\gamma\right)^2.
\end{align*}
\end{theorem}

The sufficient conditions given in Theorems~\ref{theorem_collisionAvoidanceCondForDelta} and \ref{theorem_collisionAvoidanceCondForDelta_withAcc} are likely conservative in practice.
For example, in the simulation example shown in Figure~\ref{fig_sim_2DTransConstant}, no two agents collide during the formation evolution even though the inequality \eqref{eq_collisionAvoidanceCondition_constantVelocity} does not hold.
Specifically, the left-hand side of \eqref{eq_collisionAvoidanceCondition_constantVelocity} equals 325.88, whereas the right-hand side with $\gamma=0$ equals 14.53.

\section{Conclusions}\label{section_conclusion}
This work proposed and analyzed a bearing-based approach to the problem of translational and scaling formation maneuver control in arbitrary dimensional spaces.
We proposed a variety of bearing-based formation control laws and analyzed their global formation stability.
There are several important directions for future research.
For example, in this work we assume that the information flow between any two followers is bidirectional.
In the directional case, a new notion termed \emph{bearing persistence} emerges and plays an important role in the formation stability analysis \cite{zhao2015CDC}.
Secondly, although the double-integrator dynamics can approximately model some practical physical systems, more complicated models such as nonholonomic models should be considered in the future.

\appendix

\subsection{Preliminaries to Bearing Rigidity Theory}

Some basic concepts and results in the bearing rigidity theory are revisited here. Details can be found in \cite{zhao2014TACBearing}.
For a formation $\G(p)$ with undirected graph $\G$, assign a direction to each edge in $\G$ to obtain an oriented graph.
Express the edge vector and the bearing for the $k$th directed edge in the oriented graph, respectively, as $e_{k}$ and $g_{k}\triangleq {e_{k}}/{\|e_{k}\|}$ for $k\in\{1,\dots,m\}$ where $m=|\E|$.
Define the \emph{bearing function} $F_B: \R^{dn}\rightarrow\R^{dm}$ as
$F_B(p)\triangleq [g_1^T ,\dots, g_m^T]^T$.
The \emph{bearing rigidity matrix} is defined as the Jacobian of the bearing function,
$R_B(p) \triangleq {\partial F_B(p)}/{\partial p}\in\R^{dm\times dn}$.
The bearing rigidity matrix satisfies $\rank(R_B)\le dn-d-1$ and $\myspan\{\one\otimes I_d, p\}\subseteq \Null(R_B)$ \cite{zhao2014TACBearing}.
Let $\delta p$ be a variation of $p$.
If $R_B(p)\delta p=0$, then $\delta p$ is called an \emph{infinitesimal bearing motion} of $\G(p)$.
A formation always has two kinds of \emph{trivial} infinitesimal bearing motions: translation and scaling of the entire formation.
A formation is called \emph{infinitesimally bearing rigid} if all the infinitesimal bearing motions are trivial.
The infinitesimal bearing rigidity has the following important properties.

\begin{theorem}[\cite{zhao2014TACBearing}]\label{theorem_IBR_NSCondition}
    The following statements are equivalent:
    \begin{enumerate}[(a)]
    \item $\G(p)$ is {infinitesimally bearing rigid};
    \item $\G(p)$ can be uniquely determined up to a translational and scaling factor by the inter-neighbor bearings $\{g_{ij}\}_{(i,j)\in\E}$;
    \item $\rank(R_B)=dn-d-1$;
    \item $\Null(R_B)=\myspan\{\one_n\otimes I_d, p\}$.
    \end{enumerate}
\end{theorem}

{\small
\section*{Acknowledgements}
The work presented here has been supported by the Israel Science Foundation (grant No. 1490/13).
}

\bibliography{myOwnPub,zsyReferenceAll} 

\begin{thebibliography}{10}

\bibitem{Oh2015Automatica}
K.-K. Oh, M.-C. Park, and H.-S. Ahn, ``A survey of multi-agent formation
  control,'' {\em Automatica}, vol.~53, pp.~424--440, March 2015.

\bibitem{RenWei2007SCL}
W.~Ren, ``Multi-vehicle consensus with a time-varying reference state,'' {\em
  Systems \& Control Letters}, vol.~56, pp.~474--483, 2007.

\bibitem{SunZhiyongMSC2015}
Z.~Sun and B.~D.~O. Anderson, ``Rigid formation control with prescribed
  orientation,'' in {\em Proceedings of the 2015 IEEE Multi-Conference on
  Systems and Control}, pp.~639--645, September 2015.

\bibitem{bishopconf2011rigid}
A.~N. Bishop, ``Stabilization of rigid formations with direction-only
  constraints,'' in {\em Proceedings of the 50th IEEE Conference on Decision
  and Control and European Control Conference}, (Orlando, FL, USA),
  pp.~746--752, December 2011.

\bibitem{Eren2012IJC}
T.~Eren, ``Formation shape control based on bearing rigidity,'' {\em
  International Journal of Control}, vol.~85, no.~9, pp.~1361--1379, 2012.

\bibitem{zhao2015ECC}
S.~Zhao and D.~Zelazo, ``Bearing-based distributed control and estimation in
  multi-agent systems,'' in {\em Proceedings of the 2015 European Control
  Conference}, (Linz, Austria), pp.~2207--2212, July 2015.

\bibitem{zelazo2014SE2Rigidity}
D.~Zelazo, A.~Franchi, and P.~R. Giordano, ``Rigidity theory in {SE(2)} for
  unscaled relative position estimation using only bearing measurements,'' in
  {\em Proceedings of the 2014 European Control Conference}, (Strasbourgh,
  France), pp.~2703--2708, June 2014.

\bibitem{nima2009TR}
N.~Moshtagh, N.~Michael, A.~Jadbabaie, and K.~Daniilidis, ``Vision-based,
  distributed control laws for motion coordination of nonholonomic robots,''
  {\em IEEE Transactions on Robotics}, vol.~25, pp.~851--860, August 2009.

\bibitem{TronVisionFormation}
R.~Tron, J.~Thomas, G.~Loianno, J.~Polin, V.~Kumar, and K.~Daniilidis,
  ``Vision-based formation control of aerial vehicles,'' in {\em Workshop on
  Distributed Control and Estimation for Robotic Vehicle Networks}, 2014.

\bibitem{bishop2010SCL}
M.~Basiri, A.~N. Bishop, and P.~Jensfelt, ``Distributed control of triangular
  formations with angle-only constraints,'' {\em Systems \& Control Letters},
  vol.~59, pp.~147--154, 2010.

\bibitem{Franchi2012IJRR}
A.~Franchi, C.~Masone, V.~Grabe, M.~Ryll, H.~H. B{u}lthoff, and P.~R. Giordano,
  ``Modeling and control of {UAV} bearing formations with bilateral high-level
  steering,'' {\em The International Journal of Robotics Research}, vol.~31,
  no.~12, pp.~1504--1525, 2012.

\bibitem{Cornejo2013IJRR}
A.~Cornejo, A.~J. Lynch, E.~Fudge, S.~Bilstein, M.~Khabbazian, and J.~McLurkin,
  ``Scale-free coordinates for multi-robot systems with bearing-only sensors,''
  {\em The International Journal of Robotics Research}, vol.~32, no.~12,
  pp.~1459--1474, 2013.

\bibitem{ZhengRonghao2013SCL}
R.~Zheng and D.~Sun, ``Rendezvous of unicycles: A bearings-only and perimeter
  shortening approach,'' {\em Systems \& Control Letters}, vol.~62,
  pp.~401--407, May 2013.

\bibitem{zhao2013SCLDistribued}
S.~Zhao, F.~Lin, K.~Peng, B.~M. Chen, and T.~H. Lee, ``Distributed control of
  angle-constrained cyclic formations using bearing-only measurements,'' {\em
  Systems \& Control Letters}, vol.~63, no.~1, pp.~12--24, 2014.

\bibitem{Eric2014ACC}
E.~Schoof, A.~Chapman, and M.~Mesbahi, ``Bearing-compass formation control: A
  human-swarm interaction perspective,'' in {\em Proceedings of the 2014
  American Control Conference}, (Portland, USA), pp.~3881--3886, June 2014.

\bibitem{zhao2014TACBearing}
S.~Zhao and D.~Zelazo, ``Bearing rigidity and almost global bearing-only
  formation stabilization,'' {\em IEEE Transactions on Automatic Control},
  vol.~pp, no.~99, pp.~1--1, 2015.
\newblock (Early Access).

\bibitem{Coogan2012Scale}
S.~Coogan and M.~Arcak, ``Scaling the size of a formation using relative
  position feedback,'' {\em Automatica}, vol.~48, pp.~2677--2685, October 2012.

\bibitem{Park2014IJRNC}
M.-C. Park, K.~Jeong, and H.-S. Ahn, ``Formation stabilization and resizing
  based on the control of inter-agent distances,'' {\em International Journal
  of Robust and Nonlinear Control}, 2014.
\newblock doi:10.1002/rnc.3212.

\bibitem{LinZhiyun2014TAC}
Z.~Lin, L.~Wang, Z.~Han, and M.~Fu, ``Distributed formation control of
  multi-agent systems using complex laplacian,'' {\em IEEE Transactions on
  Automatic Control}, vol.~59, no.~7, pp.~1765--1777, 2014.

\bibitem{zhao2015NetLocalization}
S.~Zhao and D.~Zelazo, ``Localizability and distributed protocols for
  bearing-based network localization in arbitrary dimensions,'' 2015.
\newblock submitted to Automatica (preprint arXiv:1502.00154).

\bibitem{zhao2015MSC}
S.~Zhao and D.~Zelazo, ``Bearing-based formation maneuvering,'' in {\em
  Proceedings of the 2015 IEEE Multi-Conference on Systems and Control},
  (Sydney, Australia), pp.~658--663, September 2015.

\bibitem{Yang2010Automatica}
P.~Yang, R.~A. Freeman, G.~J. Gordon, K.~M. Lynch, S.~S. Srinivasa, and
  R.~Sukthankar, ``Decentralized estimation and control of graph connectivity
  for mobile sensor networks,'' {\em Automatica}, vol.~46, no.~2, pp.~390--396,
  2010.

\bibitem{RenWei2008TAC}
W.~Ren, ``On consensus algorithms for double-integrator dynamics,'' {\em IEEE
  Transactions on Automatic Control}, vol.~53, no.~6, pp.~1503--1509, 2008.

\bibitem{YuWenwu2010Automatica}
W.~Yu, G.~Chen, and M.~Cao, ``Some necessary and sufficient conditions for
  second-order consensus in multi-agent dynamical systems,'' {\em Automatica},
  vol.~46, pp.~1089--1095, 2010.

\bibitem{MengZiyang2013SCL}
Z.~Meng, Z.~Zhao, and Z.~Lin, ``On global leader-following consensus of
  identical linear dynamic systems subject to actuator saturation,'' {\em
  Systems \& Control Letters}, vol.~62, pp.~132--142, 2013.

\bibitem{das2002vision}
A.~K. Das, R.~Fierro, V.~Kumar, J.~P. Ostrowski, J.~Spletzer, and C.~J. Taylor,
  ``A vision-based formation control framework,'' {\em IEEE Transactions on
  Robotics and Automation}, vol.~18, pp.~813--825, October 2002.

\bibitem{Andreasson2014TAC}
M.~Andreasson, D.~V. Dimarogonas, H.~Sandberg, and K.~H. Johansson,
  ``Distributed control of networked dynamical systems: static feedback,
  integral action and consensus,'' {\em IEEE Transactions on Automatic
  Control}, vol.~59, no.~7, pp.~1750--1764, 2014.

\bibitem{oshri2015ECC}
O.~Rozenheck, S.~Zhao, and D.~Zelazo, ``A proportional-integral controller for
  distance-based formation tracking,'' in {\em Proceedings of the 2015 European
  Control Conference}, (Linz, Austria), pp.~1687--1692, July 2015.

\bibitem{Han2015MSC}
Z.~Han, Z.~Lin, Z.~Chen, and M.~Fu, ``Formation maneuvering with collision
  avoidance and connectivity maintenance,'' in {\em Proceedings of the 2015
  IEEE Multi-Conference on Systems and Control}, pp.~652--657, September 2015.

\bibitem{zhao2015CDC}
S.~Zhao and D.~Zelazo, ``Bearing-based formation stabilization with directed
  interaction topologies,'' in {\em Proceedings of the 54th IEEE Conference on
  Decision and Control}, (Osaka, Japan), December 2015.
\newblock to appear (preprint arXiv:1508.06961).

\end{thebibliography}
\bibliographystyle{ieeetr}

\begin{IEEEbiography}[{\includegraphics[width=1in,height=1.25in,clip,keepaspectratio]{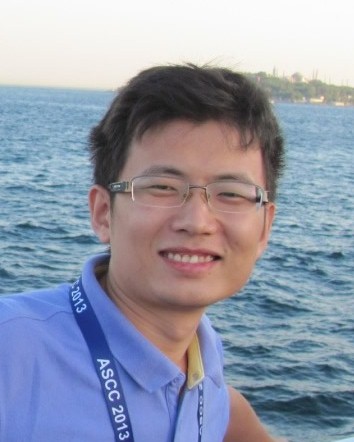}}]{Shiyu Zhao} received the B.Eng. and M.Eng. degrees from Beijing University of Aeronautics and Astronautics, Beijing, China, in 2006 and 2009, respectively. He got the Ph.D. degree in electrical engineering from the National University of Singapore in 2014.

He is a Postdoctoral Research Associate in the Department of Mechanical Engineering at the University of California, Riverside, CA, USA.
From 2014 to 2015, he was a Postdoctoral Research Associate in the Faculty of Aerospace Engineering at the Technion - Israel Institute of Technology, Haifa, Israel.
His research interests lie in distributed control and estimation over networked dynamical systems.
\end{IEEEbiography}

\begin{IEEEbiography}[{\includegraphics[width=1in,height=1.25in,clip,keepaspectratio]{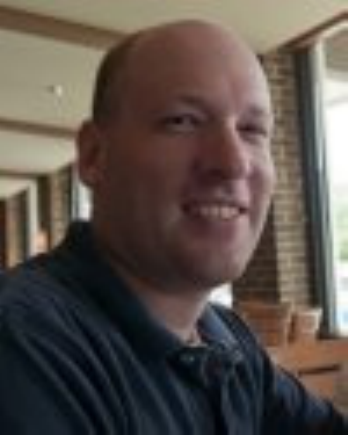}}]{Daniel Zelazo} received the B.Sc. and M.Eng. degrees in electrical engineering from the Massachusetts
Institute of Technology, Cambridge, MA, USA, in 1999 and 2001, respectively. He got the Ph.D. degree in aeronautics and astronautics from the University of Washington, Seattle, WA, USA, in 2009.

He is an Assistant Professor of aerospace engineering at the
Technion - Israel Institute of Technology, Haifa, Israel.
From 2010 to 2012, he was a Postdoctoral Research Associate and Lecturer at the Institute for Systems Theory \& Automatic Control, University of Stuttgart, Stuttgart, Germany.
His research interests include topics related to multi-agent systems, optimization, and graph
theory.
\end{IEEEbiography}

\end{document}